\documentclass[12pt,reqno]{amsart}
\usepackage{latexsym}
\usepackage{graphics}
\usepackage{hyperref, enumerate}
\usepackage{soul}
\usepackage{graphicx} 
\usepackage{amsmath}
\usepackage{amssymb,amsthm}
\usepackage{verbatim}
\usepackage[usenames,dvipsnames,svgnames,table]{xcolor}
\usepackage{cancel}
\usepackage{nicefrac}
\usepackage{doi}

\def\*{*}
\def\tr{\mathrm{Tr}}

\def\Tr{\mathop {\rm Tr}}

\newcommand{\cok}{\overline{\mathrm{co}}(\cK_{+})}
\newcommand{\coh}{\mathrm{co}(\cK_{+})}
 
\def\ep{\epsilon}
\def\q{q}
\def\be{\begin{equation}}
\def\beq{\begin{eqnarray}}
\def\beqs{\begin{eqnarray*}}
\def\ee{\end{equation}}
\def\eeq{\end{eqnarray}}
\def\eeqs{\end{eqnarray*}}
\def\eqref#1{(\ref{#1})}
\def\lra#1{\langle #1 \rangle}
\def\lrp#1{\left( #1 \right)}

\def\abs#1{\left\vert #1\right\vert}

\def\A{\mathfrak{A}}

\def\nn{\nonumber\\}
\def\N{\mathbb{N}}
\def\Z{\mathbb{Z}}
\def\R{\mathbb{R}}
\def\C{\mathbb C}
\def\bE{\mathbb E}

\def\qsp#1{\quad\text{#1}\quad}

\def\cH{\mathcal{H}}
\def\cK{\mathcal{K}}

\def\cA{\mathcal{A}}
\def\cI{\mathcal{I}}

\def\:{{:}}

\def\one{\mathbf{1}}
\def\cN{\mathcal{N}}

\renewcommand{\leq}{\leqslant}
\renewcommand{\geq}{\geqslant}

\newtheorem{thm}{Theorem}[section]
\newtheorem{theorem}[thm]{Theorem}

\newtheorem{corollary}[thm]{Corollary}

\newtheorem{lemma}[thm]{Lemma}

\newtheorem{proposition}[thm]{Proposition}
\newtheorem{definition}[thm]{Definition}

\theoremstyle{remark}
\newtheorem{remark}[thm]{\bf Remark}
\newtheorem{question}{\bf Question}

\usepackage{graphics}
\title{Reflection Positive Doubles}
\author[Arthur Jaffe]{Arthur Jaffe}
\address{Harvard University\\
Cambridge, MA 02138, USA}
\email{arthur\_jaffe@harvard.edu}

\author[Bas Janssens]{Bas Janssens}
\address{Universiteit Utrecht\\
3584 CD Utrecht, The Netherlands}
\email{B.Janssens@uu.nl}
\begin{document}
\begin{abstract}
Here we introduce \emph{reflection positive doubles},
a general framework for reflection positivity, covering 
a wide variety of systems 
in statistical physics and 
quantum field theory.
These systems may be bosonic, fermionic, or parafermionic in nature.
Within the framework of reflection positive doubles, we give necessary and  sufficient conditions for reflection positivity.
We use a reflection-invariant cone to implement our construction.
Our characterization allows for a direct interpretation in terms of 
coupling constants, making it easy to check in concrete situations. We illustrate our methods with numerous examples.
\end{abstract}
\maketitle
\thispagestyle{empty}
\tableofcontents
\setcounter{equation}{0}
\section{Introduction}
There is amazing synergy among a number of developments in operator algebra theory, quantum field theory, and statistical physics that first emerged in the 1960's and 1970's.  At the time several of these advances appeared independently, but  we now understand  them as part of a larger picture.  Their interrelation may well lead to further deep insights. 

The advances we think of include, on the side of mathematics, the Tomita-Takesaki theory 
for von Neumann algebras \cite{Tomita1967,Takesaki1970},  the  $j$-positive states of Woronowicz \cite{Woronowicz1972}, and
the self-dual cones of
Araki, Connes, and Haagerup \cite{Araki1974,Connes1974,Haagerup1975}.
On the side of physics, they include
the reflection positivity property discovered by Osterwalder and Schrader for classical fields \cite{OS1,OSEuclideanFields,OS2}. 
In perspective, we now understand how these apparently different ideas overlap as central themes in mathematics and physics.

Following the ground-breaking mathematical work of Tomita and Takesaki, 
and motivated by the 
work of Powers and St{\o}rmer on states of the CAR algebra \cite{PowersStormer1970},
Woronowicz introduced  $j$-\emph{positivity}.
In the context of a subalgebra $\A_+ \subseteq \A$ that is interchanged 
with its commutant $\A_- \subseteq \A$ by an 
antilinear homomorphism  $j \colon \A \rightarrow \A$,
this means that a state $\omega \colon \A \rightarrow \C$ 
satisfies 
\be
\omega(j(A)A) \geq 0 \quad \text{for all} \quad A\in \A_+,\quad
\quad \text{(bosonic case)}.
\ee
For a $\sigma$-finite von Neumann algebra $\A_+ \subseteq B(\cH)$ 
with modular involution $j(A) = JAJ$, 
Araki and Connes independently realized that the 
normal,
$j$-positive states constitute a 
\emph{self-dual cone} in~$\cH$.
Connes proved that
$\sigma$-finite von Neumann algebras 
are 
classified up to isomorphism by their self-dual cone; Haagerup's
subsequent generalization of these results to the 
non $\sigma$-finite case
led to the abstract formulation of
Tomita-Takesaki Theory, which proved to have a 
lasting impact on operator algebra theory.

On the side of mathematical physics, Osterwalder and Schrader formulated the idea of reflection positivity in the context of the Green's functions for the statistical mechanics of classical fields.  
In hindsight, one understands that this idea is closely related to $j$-positivity, with 
$j$ replaced by the time reflection $\Theta$. In this context, reflection positivity is expressed as
\be\label{General-RP-1}
S(\Theta(F)F)\geqslant 0\,,\qquad \text{(bosonic case)}
\ee 
where $S$ denotes the Schwinger functional defined on an algebra of test functions $F$, and the positivity \eqref{General-RP-1} holds on a subalgebra of functions supported at positive time. 

For fermionic systems, the reflection positivity condition  can be  given in terms 
of an \emph{antihomomorphism} $\Theta_{a}$ on the 
$\Z_2$-graded algebra $\A$ of test functions. Here, the positive time 
subalgebra $\A_+$ \emph{supercommutes} with the negative time subalgebra 
$\A_-$.  Reflection positivity for a gauge-invariant functional $S$ means that
$S(\Theta_{a}(F)F) \geq 0$ for all $F\in \A_+$.
To connect with $j$-positivity, 
note that as the algebra is super-commutative, 
$\Theta(F) = i^{-\abs{F}^2}\Theta_{a}(F)$ is an antilinear \emph{homomorphism}.
Here $\abs{F}$ denotes the $\Z_2$-degree of $F$, which is 0 for even 
and $1$ for odd elements.  In terms of $\Theta$, the reflection positivity condition 
becomes
\be\label{eq:fermion}
i^{\abs{F}^2}S(\Theta(F)F)\geqslant 0\,.\qquad \text{(fermionic case)}
\ee

In our previous work \cite{JaffeJanssens2016}, we discovered that 
\eqref{eq:fermion} gives the correct formulation of reflection positivity 
for Majorana fermions, where the 
$\Z_2$-graded algebra is no longer super-commutative.  

In this paper we generalize this condition to the case of a neutral functional $S$ on a $\Z_{p}$-graded (rather than $\Z_{2}$ graded) algebra. We find that the reflection positivity condition becomes
\be\label{General-RP-3}
\zeta^{\abs{F}^2}S(\Theta(F)F)\geqslant0\;, \qquad\text{(general case)}
%
\ee
with $F$ a homogeneous element of $\A_+$ of degree $\abs{F}$, and with $\zeta$ an appropriate $2p^{\rm th}$ root of unity.

\subsection{Applications to Mathematical Physics}
The importance of the Osterwalder-Schrader construction stems from the fact that the Hilbert space of \textit{every} known scalar quantum theory arises as a quantization defined by this framework; quantum theories with fermions or gauge fields arise as generalizations of this approach.  

Another early application of reflection positivity  was its use in papers of Glimm, Jaffe, and Spencer to give the first proof in that interacting, non-linear quantum fields  satisfy the Wightman axioms \cite{GJS74}.  Shortly  afterward, these authors used reflection positivity to give the first mathematical proof that discrete symmetry breaking and phase transitions exist in certain quantum field theories \cite{GlimmJaffeSpencer-PhaseTransitions1975}.  

 The analysis of Tomita-Takesaki theory led Bisognano and Wichmann to identify the TCP reflection $\Theta$ in quantum field theory  with a specific case of the Tomita reflection $j$ defined for wedge shaped regions \cite{BisognanoWichmann1975}.  Sewell recognized that the Bisognano-Wichmann theory yields an interpretation of Hawking radiation from black holes \cite{Sewell1980}. Hislop and Longo analyzed the modular structure of double cone algebras
 in great detail \cite{HilsopLongo}. 
Also much work of Borchers, Buchholz, Fredenhagen, Rehren, Summers, and others has been devoted to aspects of relations between local quantum field theory and Tomita-Takesaki theory.   

Turning from quantum mechanics to lattice statistical physics, reflection positivity was established by Osterwalder and Seiler for 
lattice gauge theory \cite{O76}, as well as for the super-commutative, $\Z_2$-graded algebra  appearing in lattice QCD \cite{OS78}.
Reflection positivity was also central in the first proof of continuous symmetry breaking in lattice systems, through the proof and use of ``infrared bounds'' by Fr\"ohlich, Simon, and Spencer  \cite{FroehlichSimonSpencer}.  Reflection positivity also led to many results by Fr\"ohlich, Israel, Lieb, and Simon   \cite{FILS78}, for bosonic quantum systems in classical and quantum statistical physics.

\subsection{Positive Cones, Twisted Products}
The analysis of reflection positivity in the papers  \cite{OSEuclideanFields,OS78,FILS78}  used a \emph{cone} of reflection positive elements.  This cone mirrors the positivity conditions \eqref{General-RP-1}--\eqref{eq:fermion} for bosonic and fermionic systems. The cone consists of elements 
\be\label{Cone-1} 
	\Theta(A)A, \quad A \in \A_+\,,\qquad\text{(bosonic case)}
\ee
where $\A_+$ is the algebra of observables on one side of the reflection plane.
In terms of the antilinear homomorphism $\Theta$, the cones of Oster\-walder, Schrader,  and Seiler
consist of elements of the form
\be\label{Cone-2} 
	i^{\abs{A}^2}\Theta(A)A, \quad A\in \A_+\,. \qquad \text{(fermionic case)}
\ee 
In hindsight, it is clear that these cones are closely related to the self-dual cone 
of Araki, Connes and Haagerup.

In this paper, we isolate a minimal framework for reflection positivity,
covering a variety of different and useful examples, including the above.
In brief, we work with a $\Z_p$-graded, locally convex algebra $\A$, equipped with 
an antilinear homomorphism $\Theta \colon \A \rightarrow \A$
called the \emph{reflection}.

Let $\A_+ \subseteq \A$ be a graded subalgebra, 
and $\A_- = \Theta(\A_+)$. We assume that $\Theta$ squares to the identity and inverts the grading. Then $\A$ is called the \emph{q-double} of $\A_+$
if the linear span of $\A_-\A_+$ is dense in $\A$, and if $\A_+$ \emph{paracommutes} with $\A_-$, meaning that
	\be\label{eq:exchangeintro}
		A_{-}A_{+} 
		= q^{\abs{A_{-}}\abs{A_{+}}}A_{+}A_{-} \;,
	\ee
for homogeneous $A_{\pm} \in \A_{\pm}$. Here
$q = e^{2\pi i/p}$ is a $p^{\mathrm{th}}$ root of unity, and $\abs{A}$ 
denotes the degree of $A$ in $\Z_p$.    The case $p=1$ describes bosons, the case $p=2$ describes 
fermions, and the case $p>2$ corresponds to parafermions. We give more details in \S \ref{sec:section2}.
The generalization of \eqref{Cone-1}--\eqref{Cone-2}   is the 
\emph{reflection positive cone}  $\cK_+$ with elements 
\be\label{Cone-3}
\zeta^{\abs{A}^2}\Theta(A)A\,. \qquad\text{(general case)}
\ee
Here  $A$ a homogeneous element of $\A_+$, and $\zeta$ is a square root of $q$, with  $\zeta^{p^{2}}=1$.
The cone $\cK_+$ is closed under multiplication, 
and point-wise fixed by the reflection $\Theta$. We find it useful to consider the expression \eqref{Cone-3} as a specialization of a reflection-invariant, twisted product,
	\be\label{TwistedProduct}
	\Theta(A)\circ A
	=
	\zeta^{\abs{A}^2}\Theta(A)A\;.
	\ee

Parafermionic commutation relations were proposed 
in field theory by Green \cite{Green1953}.
They are closely tied to representations of the braid group, which
lead to a variety of different statistics, see for 
example~\cite{FroehlichGabbiani1990}. 
Recently, 
Fendley gave a parafermionic representation of Baxter's clock hamiltonian
\cite{Fendley12,Fendley14}. In \cite{Jaffe-Pedrocchi2015-2},
Jaffe and Pedrocchi
gave sufficient conditions for reflection positivity 
on the $0$-graded part of the parafermion algebra,
and used this to study topological order \cite{JaffePedrocchi2014}.
Jaffe and Liu found a geometric interpretation of reflection positivity in the framework of 
planar para algebras, relating reflection positivity in that case to $C^{*}$ positivity \cite{JL2016}. Their proof uses an elegant pictorial interpretation for the twisted product $\Theta(A)\circ A$ in \eqref{TwistedProduct}, as an interpolation between $\Theta(A)A$ and $A\Theta(A)$.

Recently the ideas from Tomita-Takesaki theory have been used by workers in string theory, for instance in the analysis of the black hole complementarity radiation, see  \cite{PS2013}.  It is tempting to conjecture that string theory representations of black-hole partition functions  of the form $Z_{\rm BH}= |Z_{\rm top}|^{2}$,  proposed in \cite{OSV03,GSX,Pestun13}, have an origin in (and an explanation based on)  reflection positivity.

\subsection{Overview of the Present Paper}
Let $\tau \colon \A \rightarrow \C$ be a continuous, reflection positive `background functional', 
and let $H\in \A$ be a reflection invariant element of degree zero.
The main problem is to determine necessary and sufficient conditions on $H$
for the functional 
\be\label{eq:overviewBoltzmann}
\tau_{H}(A) = \tau(Ae^{-H})
\ee
to be reflection positive.
In the case of statistical physics, $H$ is a Hamiltonian and $\tau_{H}$ defines the Boltzmann functional for the system. In the case of functional integrals for quantum theory, $H$ is a perturbation of the action. 

In \S \ref{sec:section2} we give the basic definitions.  In \S\ref{sec:applications}, we apply this general setting to  
a variety of different situations: Tomita-Takesaki theory and von Neumann algebras (describing bosonic systems),
Grassmann algebras (describing
fermionic classical systems), and Clifford algebras and CAR algebras (describing fermionic quantum systems).
Finally, we introduce the \emph{parafermion algebra} and the \emph{CPR algebra},
the analogues of Clifford and CAR algebras for parafermions.

Let us point out that in our general setting, we do \emph{not} assume that $\A$ is a $*$-algebra,
nor that $\tau$ is a state.
This allows our framework to cover cases such as Berezin integration on Grassmann 
algebras \cite{Berezin}, and neutral complex fields in the sense of~\cite{JaffeJaekelMartinez,JaffeJaekelMartinez2014b}.
 
In \S\ref{Sect:SforRP} and \S\ref{Sect:N-SforRP}, we return to the problem of determining reflection positivity of $\tau_{H}$
in the general setting. 
Our first main result is
Theorem~\ref{Thm:sufficient}; 
the Boltzmann functional $\tau_{H}$ is reflection positive if 
$H$ allows a decomposition
\[
	H = H_- + H_0 + H_+\,,
\]
where $H_{+}$ is in $\A_+$, $H_-$ is the reflection of $H_+$ in $\A_-$, and $-H_0$ is in the closure 
of the convex hull of $\cK_+$. Although the decomposition is familiar, 
the result is new for Hamiltonians of the generality that we study here.

Our second major result, Theorem \ref{thm:necessary}, states that these conditions are not only sufficient, 
but also necessary,
under additional factorization and nondegeneracy 
hypotheses on the `background' functional~$\tau$.
These extra assumptions are reasonable, and generally assumed, 
in the framework of statistical physics.

In our third main result, Theorem~\ref{Thm:neccandsuff}, we formulate necessary and sufficient conditions 
for reflection positivity in terms of the \emph{matrix of coupling constants across the 
reflection plane}. These are the coefficients $J^{0}_{IJ}$ of $H$ with respect to a 
distinguished, reflection-invariant basis $B_{IJ}$
of the zero-graded algebra $\A^0$.
This is particularly relevant in the context of 
statistical physics, where the Hamiltonian is usually given in terms of
couplings. Theorem~\ref{Thm:neccandsuff} then allows one to easily check reflection positivity 
in concrete situations. 
 
We end the paper with an extensive list of examples in the context of 
lattice statistical physics. In \S\ref{sec:latphys}, 
we discuss the lattices we use.

In \S\ref{sec:bosonicsystems}, 
we specialize our results to bosonic classical and quantum systems.
A special feature of \emph{classical} systems is that the lattice can contain fixed points under the reflection.
We exploit this to prove the following: suppose that the reflection is in one of the coordinate directions, that the lattice is rectangular, and that it has nontrivial intersection with 
the reflection plane.
Then \emph{every} reflection-invariant nearest neighbor Hamiltonian
yields a reflection positive Boltzmann measure.

In \S\ref{sec:fermionicsystems}, we specialize our results to fermionic classical and quantum systems.
In the classical case, the `background functional' is the Berezin integral
on the Grassmann algebra, and in the quantum case, it is the tracial state on the Clifford algebra 
or CAR algebra.
The results in this section generalize the examples in our previous work \cite{JaffeJanssens2016}. 

In \S\ref{sec:latticeQCD}, we apply our results in the 
context of lattice gauge theories. In particular, we give a new, gauge equivariant 
proof of reflection positivity for the expectation defined by the Wilson action. 
In contrast to the proofs  in  \cite{O76,OS78,S82,MP87}, we prove reflection positivity on the \emph{full} algebra of observables, not just on the gauge invariant subalgebra.

Finally, in \S\ref{sec:parafermionsNo2}, we give a complete characterization of reflection positivity
for parafermions. This extends the results of \cite{Jaffe-Pedrocchi2015-2} from the degree zero 
subalgebra $\A^{0}_+$ to the full algebra $\A_+$, and ties in with 
the results of Jaffe and Liu \cite{JL2016} on the geometric interpretation of 
reflection positivity for planar para algebras.
\setcounter{equation}{0}
\section{Reflection Positivity for $\Z_p$-graded Algebras}\label{sec:section2}

In this section, we introduce the basic notions needed to treat
reflection positivity in the $\Z_p$-graded setting.

\subsection{Graded Topological Algebras}

Let $p\in \N$, and let $\A$ be a  $\Z_{p}$-graded unital algebra.
The case $p=1$ corresponds to bosons, the case 
$p=2$ corresponds to fermions, and the case $p>2$ corresponds to parafermions.  
The case $p=0$ is also allowed; as $\Z_{p} = \Z/p\Z = \Z$ for $p=0$, this
corresponds to $\Z$-graded algebras.

Denote the degree (or grading) of $A\in \A$ by $\abs{A} \in \Z_{p}$, and 
denote the homogeneous part of degree $k$ by
$\A^{k} = \{A\in \A\,;\, \abs{A} = k\}$. 
The algebra $\A$ then decomposes as
\be\label{DirectSumA}
\A = \bigoplus_{k\in \Z_{p}} \A^{k}\,.
\ee 

We require that $\A$ be a locally convex topological algebra.  This means that $\A$ is a locally convex (Hausdorff) topological vector space, for which the multiplication is separately continuous. 
In the case $p=0$, we will allow algebras for which the right hand side of \eqref{DirectSumA} is dense in $\A$.

Note that $\A$ need not be a $*$-algebra.
The above setting therefore includes not only (graded)
$C^{*}$-algebras and von Neumann algebras, but also
Grassmann algebras and continuous inverse algebras.

\subsection{Reflections and $q$-Doubles.}

Reflection positivity will be defined with respect to a 
\emph{reflection} $\Theta \colon \A \rightarrow \A$.

\begin{definition}[\bf Reflections]  \label{Def:R}
A reflection $\Theta \colon \A \rightarrow \A$ is 
a continuous, anti-linear homomorphism which squares to the identity and inverts the grading.  
\end{definition}
In other words, we require that $\Theta^2 = \mathrm{Id}$, and that
$\abs{\Theta(A)}=-\abs{A}$ for homogeneous elements ${A\in\A}$.

\begin{remark}\label{rk:switch}
In the literature, RP is not always defined using an
anti-linear homomorphism $\Theta$, as we do.  
In fact, there are 4 possibilities. 
The transformation $\Theta$ may be 
linear or antilinear; it may be a homomorphism or an anti-homomorphism.
In the context of $*$-algebras, one can go back and forth between 
an antilinear homomorphism $\Theta$ and a 
linear anti-homomorphism $\widetilde{\Theta}$ 
by defining $\widetilde{\Theta}(A) := \Theta(A^*)$.
In the context of super-commutative algebras (Grassmann fermions), 
one can go back and forth between 
an antilinear homomorphism $\Theta$ and a 
linear anti-homomorphism $\Theta_{a}$ by defining
$\Theta(A) = i^{-\abs{A}^2}\Theta_{a}(A)$.
In the general setting that we describe in this paper, 
it seems that the 
anti-linear homomorphism in Definition~\ref{Def:R}
is the only option.
\end{remark}

Let $\A_{+}$  be a distinguished $\Z_{p}$-graded subalgebra of $\A$, 
and write $\A_{-} := \Theta(\A_{+})$ for its reflection.
Define $\A_{-}\A_{+} := \{A_{-}A_{+} \,{;}\, A_{\pm} \in \A_{\pm}\}$.
Let  $q$ be a complex number satisfying  $q^{p} = 1$ and $\abs{q} = 1$.
We require that $\A$ is the $q$-double of $\A_+$ in the following sense.

\begin{definition}[\bf $\boldsymbol{q}$-Double]\label{qDouble}
The algebra $\A$ is the $q$-double of $\A_{+}$ if 
\begin{enumerate}
\item{}
The linear span of $\A_{-} \A_{+}$ is dense in $\A$.
\item{}
The elements of $\A_{\pm}$ 
satisfy the para-commutation relations 
	\be\label{eq:exchange}
		A_{-}A_{+} 
		= q^{\abs{A_{-}}\abs{A_{+}}}A_{+}A_{-} \;,
		\quad \text{for}\quad A_{\pm}\in \A_{\pm}\,.
	\ee
\end{enumerate}
The $q$-double is called bosonic if $q=1$, fermionic if $q=-1$, 
and parafermionic if $q = e^{2\pi i/p}$ for $p\geq 3$.
\end{definition}

\begin{remark}
Consider the intersection $\A_+\cap \A_{-}$.
From the para-commutation relation 
\eqref{eq:exchange}, we infer that for all $A_0 \in \A_+\cap \A_{-}$ 
and $A_\pm \in \A_\pm$, we have
\[
A_0A_+ = q^{\abs{A_0}\abs{A_+}} A_+A_0\,, \quad
A_0A_- = q^{-\abs{A_0}\abs{A_-}} A_-A_0 \,.
\] 
Combined with the fact that $\A_-\A_+$ is dense in $\A$, this 
implies that $\A_+\cap \A_{-}$ is central 
in the bosonic case.
In the fermionic case, this implies that $\A_{+} \cap \A_{-}$ is 
supercentral, meaning that $\{A_0,  A\}_+ = 0$ for all 
$A_0 \in \A_+\cap \A_{-}$ and
 $A\in \A$.
Here $\{A, B\}_{+} = AB - (-1)^{\abs{A}\abs{B}} BA$ is the graded commutator 
on $\A$.
\end{remark}

\subsection{Functionals and Reflection Positivity}

We define reflection positivity for \emph{neutral} functionals. 

\begin{definition}[\bf Neutral Functionals]
A functional $\varrho \colon \A \rightarrow \C$ is neutral 
if $\varrho(A) = 0$ whenever $\abs{A} \neq 0$.
\end{definition}
Neutral functionals on $\A$ are determined by their restriction to $\A^0$.

\begin{definition}[\bf Reflection Invariance]\label{Def:RI}
A functional $\varrho\colon \A \rightarrow \C$ is  
reflection invariant if $\varrho(\Theta(A)) = \overline{\varrho(A)}$
for all $A\in \A$.
\end{definition}

For a given $q$ of modulus 1, let $\zeta$ be
a complex number  with
	\be
	q=\zeta^{2}\;,
	\quad\text{and}\quad
	\zeta^{p^{2}}=1\;.
	\ee
Since $\zeta^{(k+p)^2} = \zeta^{k^2}$, 
the expression $\zeta^{k^2}$ is well defined for $k\in \Z_{p}$.

\begin{definition}[\bf Sesquilinear Form on $\A_{+}$]\label{Def:SesqForm}
Let $\varrho$ be a neutral functional on $\A$. 
Then $\lra{A, B}_{\Theta,\varrho, \zeta}$ is the  
sesquilinear form on $\A_+$, with
\be\label{eq:defscalarproduct}
\lra{A, B}_{\Theta,\varrho, \zeta} 
= \zeta^{\abs{A}^{2}}\varrho(\Theta(A)B)
\ee 
for homogeneous $A,B\in\A_{+}$.
\end{definition}

Since $\varrho$ is neutral, \eqref{eq:defscalarproduct} is nonzero only
when $\abs{A} = \abs{B}$. In this case, 
$\zeta^{\abs{A}^{2}} = \zeta^{\abs{A} \,\abs{B}} = \zeta^{\abs{B}^{2}}$.  
Although the form \eqref{eq:defscalarproduct}  depends on $\zeta$, for most 
of this paper we fix the value of $\zeta$ and drop the subscript in 
$\lra{A, B}_{\Theta,\varrho}$.   

\begin{proposition}[\bf Hermitian Form on $\A_{+}$]\label{Prop:Hermiticity}
Let $\varrho$ be a continuous, neutral functional on $\A$. Then 
the sesquilinear form \eqref{eq:defscalarproduct} is hermitian on $\A_+$ if and only if 
$\varrho$ is reflection invariant.
\end{proposition}

\begin{proof} 
Let $A,B \in \A_{+}$ be homogeneous with $\abs{A} = \abs{B}$. 
Applying the para-commutation relation \eqref{eq:exchange} to
$A\in \A_{+}$ and $\Theta(B)\in \A_{-}$, we obtain
$\Theta(B)A = q^{-\abs{A}^2}A\Theta(B)$. Using this, we find
\beqs
\lra{B, A}_{\Theta,\varrho} &=&  \zeta^{\abs{A}^{2} }\varrho(\Theta(B)A) = \zeta^{-\abs{A}^{2} }\varrho(A\Theta(B))\,,\\
 \overline{\lra{A, B}}_{\Theta,\varrho} &=& 
 \overline{\zeta^{\abs{A}^{2} }\varrho(\Theta(A)B)} = 
\zeta^{-\abs{A}^{2} } \overline{\varrho(\Theta(A)B)}\,.
\eeqs
Setting $X = \Theta(A)B$ and $\Theta(X) = A \Theta(B)$, we see that
$\lra{B, A}_{\Theta,\varrho} = \overline{\lra{A, B}}_{\Theta,\varrho}$
for all $A,B \in \A_+$ if and only if $\varrho(\Theta(X)) = \overline{\varrho(X)}$
for all $X \in \A_-\A_+$.
As $\varrho$ is continuous and the linear span of $\A_-\A_+$ is dense in $\A$, the statement follows.
\end{proof}

\begin{remark}
The above argument relies heavily on the fact that $q$ is of modulus 1.
For $\abs{q} \neq 1$, reflection invariant functionals would not give rise to 
hermitian forms.
\end{remark}

\begin{definition}[\bf Reflection Positivity]\label{def:reflectionpositivity}
Let $\varrho \colon \A \rightarrow \C$ be a neutral linear 
functional. 
Then $\varrho$ is reflection positive
on $\A_{+}$ with respect to $\Theta$ if the form \eqref{eq:defscalarproduct} is positive semidefinite,
\be\label{eq:rpdef}
\lra{A, A}_{\Theta,\varrho} \geq 0\;,
\ee
for all $A \in \A_{+}$. 
\end{definition}

\begin{proposition}\label{prop:forceinvar}
Every continuous, neutral, reflection positive functional 
$\varrho\colon \A \rightarrow \C$ is reflection invariant.
\end{proposition}
\begin{proof}
Since \eqref{eq:defscalarproduct} 
is sesquilinear and positive semidefinite, it is hermitian by polarization.
The result now follows from Proposition~\ref{Prop:Hermiticity}.
\end{proof}

Note that Definition~\ref{def:reflectionpositivity} 
depends on the choice of  $\zeta$.

\begin{proposition}
Let $\varrho \colon \A \rightarrow \C$ be a neutral functional. 
Then $\varrho$ is reflection positive on $\A_{+}$ with parameter 
$\zeta$ if and only if
it is reflection positive on $\A_{-}$ with parameter~$\overline{\zeta}$.
\end{proposition}

\begin{proof}
Define the sesquilinear form $\lra{\,A\,,B\,}_{\Theta,\varrho,\overline\zeta}$
on $\A_-$ by
	\be\label{AMinusForm}
	\lra{\,A\,,B\,}_{\Theta,\varrho,\overline\zeta}
		= \zeta^{-\abs{A}^{2}}\varrho\lrp{\Theta(A)B}
	\ee
for homogeneous $A,B\in\A_{-}$.
The relation \eqref{eq:exchange} yields 
	\[
	\lra{\,A\,,B\,}_{\Theta,\varrho,\overline\zeta}
	= \zeta^{-\abs{A}^{2}}\varrho\lrp{\Theta(A)B}
	= \zeta^{\abs{A}^{2}}\varrho\lrp{B\Theta(A)}
	=\lra{\Theta(B),\Theta(A)}_{\Theta,\varrho}\;,
	\]
where both $\Theta(A), \Theta(B)\in\A_{+}$.  We infer that positivity of the form \eqref{eq:defscalarproduct} on $\A_{+}$ is equivalent to positivity of the form 
\eqref{AMinusForm} on $\A_{-}$.  
 \end{proof}

\begin{remark}
Although we singled
out the subalgebra $\A_{+}$,
all the statements in the paper
remain true if one exchanges $\A_{+}$ with $\A_{-}$, provided that one also exchanges the pair
$(q,\zeta)$ with $(\overline{q}, \overline{\zeta})$.  
\end{remark}

\begin{definition}[\bf Quantum Hilbert Space]
Let $\varrho$ be a neutral, reflection positive functional on $\A$.
Let  $\mathcal{N}\subseteq \A_{+}$ be the kernel of the positive semidefinite form 
$\lra{A,B}_{\Theta,\varrho}$.
Then the  quantum Hilbert space $\cH_{\Theta,\varrho}$
is the closure of $\A_{+}/\mathcal{N}$, with inner product induced by
the positive definite form $\lra{A,B}_{\Theta,\varrho}$.
\end{definition}

Denote the closure of $\A_+^k/\cN$ by $\cH_{\Theta,\varrho}^{k}$.
Since $\varrho$ is neutral, one has $\lra{A,B}_{\Theta,\varrho} = 0$ for 
homogeneous $A,B\in\A_{+}$
with $\abs{A}\neq \abs{B}$. 
It follows that $\cH_{\Theta,\varrho}^{k} \perp \cH_{\Theta,\varrho}^{k'}$
for $k\neq k'$, so that 
\be
\cH_{\Theta,\varrho} = \bigoplus_{k\in \Z_{p}} \cH_{\Theta,\varrho}^{k}
\ee 
is a $\Z_{p}$-graded Hilbert space.
In particular, 
in the fermionic case $p=2$, the space 
$\cH_{\Theta,\varrho}$ with the form $(A,B) = \varrho(\Theta(A)B)$ 
is a \emph{super Hilbert space}
in the sense of \cite{SuperSolutions}.

\subsection{The Twisted Product}\label{Sect:TwistedProduct}
For $A\in \A_-$ and $B\in \A_+$, we introduce the \emph{twisted product} ${A \circ B}$.
It allows for a convenient reformulation of reflection positivity, and
plays an important role in \cite{Jaffe-Pedrocchi2015-2,JaffeJanssens2016,JL2016}.

\begin{definition}[\bf Twisted Product]  \label{Def:TwistedProduct}
The twisted product is the bilinear map $\A_- \times \A_+ \rightarrow \A$
defined by  
	\be\label{eq:TwistedProduct}
		A\circ B 
		= \zeta^{\abs{A}^{2}} \,AB
		= \zeta^{\abs{B}^{2}}\,AB
	\ee
if $A\in \A_-$ and $B\in \A_+$ are homogeneous with $\abs{A} = - \abs{B}$, and by 
$B\circ A = 0$ if they are homogeneous with $\abs{A} \neq - \abs{B}$.
\end{definition}
\begin{remark}
If $\abs{B} = - \abs{A}$,
then $A\circ B = \zeta^{-\abs{A}\abs{B}}\,AB$
interpolates 
between the products $AB$ and $BA = q^{-\abs{B}\abs{A}}AB$.
If $\abs{B} \neq - \abs{A}$, then $\varrho(AB) = \varrho(BA) = 0$ for any neutral 
functional $\varrho$ on $\A$.
As we will mainly be interested in expressions of the form $\varrho(A\circ B)$, 
we have put $A\circ B = 0$.
\end{remark}

If $\varrho \colon \A \rightarrow \C$ is a neutral functional, then the sesquilinear form 
can be expressed in terms of the twisted product as
	\be
	\lra{A,B}_{\Theta,\varrho} =
	\varrho(\Theta(A)\circ B)\;, \text{ for }A,B\in\A_{+}\;.
	\ee
In particular, the reflection-positivity condition \eqref{eq:rpdef} is
 	\be\label{eq:RPincircle}
	 \varrho\lrp{\Theta(A)\circ A}\geq 0\;,
		\quad\text{for }\quad A\in\A_{+}\;.
	\ee

\subsection{The Reflection-Positive Cone}\label{sec:firstcone}

A central idea in this paper is that the reflection-positive functionals on $\A$ can be characterized 
in terms of the \emph{reflection-positive cone} $\cK_{+} \subseteq \A$. 
In the bosonic case, this idea appeared first in connection to Tomita-Takesaki modular theory \cite{Woronowicz1972,Araki1974,Connes1974,Haagerup1975}, 
and was later used in statistical physics \cite{FILS78}.  
In the fermionic case, reflection-positive cones were used in \cite{OS78}.
We extend the definition of reflection-positive cones to the 
$\Z_{p}$-graded setting as follows.

\begin{definition}\label{defcone}
The reflection-positive cone $\cK_{+}\subseteq \A^{0}$ is the set 
\[
\cK_{+}=\{\Theta(A)\circ A \,|\, A  \text{  homogeneous in }\A_{+}  \}\,.
\] 
\end{definition}
Denote the convex hull of $\cK_{+}$ by $\mathrm{co}(\cK_{+})$, 
and  denote its closure by
$
\cok
$.  
Since every element of $\cK_{+}$ is of degree zero, $\cok\subseteq \A^{0}$.
The following proposition shows that 
the set of 
continuous, neutral, reflection positive functionals on $\A$
is precisely 
the \emph{continuous dual cone} of $\cok$ in $\A^0$.          
          
\begin{proposition}\label{prop:rpclosure}
Let $\varrho \colon \A \rightarrow \C$ be a continuous, neutral, linear functional. Then
the following are equivalent:
\begin{itemize}
\item[a)] The functional $\varrho$ is reflection positive.
\item[b)] The functional $\varrho$ is nonnegative on $\cK_{+}$.
\item[c)] The functional $\varrho$ is nonnegative on $\cok$.
\end{itemize}
\end{proposition}
\begin{proof}
The equivalence of a) and b) follows from
Definition \ref{defcone} and equation
\eqref{eq:RPincircle}.
As $\cK_{+} \subseteq \cok$, we infer that c) implies b).
To show that b) implies c), note that
since $\varrho(\cK_{+})\subseteq \R^{\geq 0}$ and $\varrho$ is linear, 
the image of the convex hull of $\cK_{+}$ is contained in $\R^{\geq 0}$.
As $\varrho$ is continuous, the same holds for its closure $\cok$.
\end{proof}

\begin{proposition}
The linear span of $\cK_+$ is dense in $\A^0$.
\end{proposition}
\begin{proof}
Expanding  $\Theta(A+ B)\circ (A+B)$ and $\Theta(A+ iB)\circ (A+iB)$ for 
$A,B\in \A_+$, one finds that 
\beqs
\Theta(A)\circ B + \Theta(B)\circ A &\in& 
\phantom{i(}\cK_+ - \cK_+\\
\Theta(A)\circ B - \Theta(B)\circ A &\in& 
i(\cK_+ - \cK_+)\,.
\eeqs
Thus $\Theta(A) \circ B \in \cK_+ - \cK_+ + i\cK_+ - i \cK_+$. 
Since the linear span of $\A_-\A_+$ is dense in $\A$, the span of 
$\Theta(\A_+)\circ \A_+$ is dense in $\A^0$.
\end{proof}

\subsection{Boltzmann Functionals}

In physical applications, the relevant functionals $\tau_{H}$ are often perturbations of a fixed
`background' functional $\tau$ by an operator $H$. 

Let $\tau \colon \A \rightarrow \C$ be a continuous, neutral,
reflection positive functional on $\A$. 
Let $H\in \A$ be a reflec\-tion-invariant operator of degree zero,
\be
\Theta(H) = H,  \quad H \in \A^{0}\,.
\ee
If the exponential series for $e^{-H}$ converges, then
the
\emph{Boltzmann functional} 
$
\tau_{H} \colon \A \rightarrow \C
$ 
is defined by
\be
\tau_{H}(A) = \tau(A\,e^{- H})\,.
\ee

\begin{proposition}
The Boltzmann functional $\tau_{H}$ is continuous, neutral, and reflection invariant.
\end{proposition}
\begin{proof}
It is continuous since both $A \mapsto Ae^{-H}$ and
$A \mapsto \tau(A)$ are continuous. It is neutral
since $H$ is of degree zero and $\tau$ is neutral, and it
is reflection invariant since both $H$ and $\tau$ are
reflection invariant. (This holds for $\tau$ by Proposition~\ref{prop:forceinvar}).
\end{proof}



\begin{remark} \label{Remark:Constant}
If $\tau_{H}$ is reflection positive
for $H\in \A$, then it is reflection positive 
for every shift $H'= H + \Delta I$ by a real number $\Delta \in \R$,
since $\tau_{H'} = e^{-\Delta} \,\tau_{H}$.
\end{remark}

In applications to statistical physics, $H$ represents the Hamiltonian of the system, 
which is usually a hermitian operator in a $*$-algebra.
Note however that we do not require $A$ to be a $*$-algebra, 
nor that $H$ be hermitian.
This is important for applications in quantum field theory, where $H$ represents the action.

\subsection{Factorization}\label{sec:factorization}

The central question is to determine whether or not $\tau_{H}$
is reflection positive in terms of tractable conditions on $H$.
The first step is to show reflection positivity of the `background' 
functional~$\tau$. This can often be 
done with the help of the following factorization criterion, expressing that 
$\A_{-}$ and $\A_{+}$ are independent under~$\tau$.

\begin{definition}[\bf Factorizing Functionals]
Let $\tau \colon \A \rightarrow \C$ be a continuous, neutral, reflection invariant 
functional. Then $\tau$ is factorizing if there exists 
a neutral, continuous functional $\tau_{+}$ on $\A_{+}$ such that 
\[
\tau(\Theta(A)\circ B) = 
\overline{\tau_{+}(A)}\,
\tau_{+}(B)\,,
\]
for all $A,B \in \A_{+}$ with $\abs{A} = \abs{B}$.
\end{definition}
Since $\tau$ is reflection invariant, this is equivalent to
\be\label{eq:altfactor}
\tau(A\circ B) = 
\tau_{-}(A)\,\tau_{+}(B) \quad \text{for all} \quad A \in \A_{-},\, B \in \A_{+} \,,
\ee
where 
$\tau_{-}(B) := \overline{\tau_{+}(\Theta(B))}$ for $B\in \A_{-}$.
Since the span of $\A_-\A_+$ is dense in $\A$, a factorizing functional 
$\tau$ is uniquely determined by 
$\tau_+\colon \A_+ \rightarrow \C$.

\begin{proposition}\label{invimppos}
Every factorizing functional $\tau \colon \A \rightarrow \C$
is reflection positive.
\end{proposition}
\begin{proof}
If $A\in \A_{+}$ is homogeneous, 
$
	\tau(\Theta(A)\circ A) = \overline{\tau_{+}(A)}\tau_{+}(A) \geq 0
$.
\end{proof}

\subsection{Strictly Positive Functionals}\label{sec:sharpsec}

The notion of \emph{strictly positive functionals} 
generalizes the notion of faithful states to algebras 
without involution, such as Grassmann algebras.
Let 
$\sharp \colon \A_{+} \rightarrow \A_{+}$ be a continuous,
antilinear map which inverts the grading, $\abs{A^{\sharp}} = -\abs{A}$.

\begin{definition}\label{def:strictpos}
The functional $\tau_{+} \colon \A_{+} \rightarrow \C$ is 
strictly positive with respect to $\sharp$
if 
\[\tau_{+}(A^{\sharp} A) >0\]
for all nonzero $A\in \A_{+}$. 
\end{definition}

If $\A$ is a $*$-algebra and $\sharp$ is the $*$-involution, 
then $\tau_+$ is strictly positive if and only if it is a (not necessarily normalized) 
faithful state.

\begin{remark}
The existence of a 
strictly positive, factorizing functional $\tau \colon \A \rightarrow \C$ implies that
$\A_{-} \cap \A_{+} = \C \one$, since the restriction of $\tau$ to $\A_{-} \cap \A_{+}$ 
must have trivial kernel.
\end{remark}

\setcounter{equation}{0}
\section{Applications}\label{sec:applications}

The above setting 
is motivated by a large number of applications and examples
in both mathematics and physics. 
Before we continue our investigation into reflection positivity 
of the Boltzmann functionals $\tau_H$, we pause to outline a number of
situations that fit into the general scheme outlined in \S\ref{sec:section2}. For more detailed applications in the context of lattice statistical physics, 
we refer to~\S\ref{sec:latphys}--\ref{sec:parafermionsNo2}.

\subsection{Tensor Products}\label{sec:tensprod}

Let $\Theta_+ \colon \A_+ \rightarrow \A_-$ be an antilinear isomorphism
of von Neumann algebras, and let 
$\A = \A_{-} \otimes \A_{+}$ be the (spatial) tensor product of $\A_-$ and $\A_+$.
Then $\Theta(A \otimes B) = \Theta_+(B) \otimes \Theta^{-1}_+(A)$
defines a reflection on $\A$, and $\A$ is a bosonic
$q$-double of $\A_+$.
Reflection positivity in this setting was studied by Woronowicz under the name 
$j$-positivity \cite{Woronowicz1972}.

%

If $\tau_{+}$ is a state on $\A_{+}$,
then 
the induced state $\tau(\Theta_+(A) \otimes B) = \overline{\tau_{+}(A)}\tau_{+}(B)$ 
on $\A$ is factorizing, and hence reflection positive. 
If $\sharp$ is the involution on $\A_+$, then $\tau_+$ is strictly positive 
in the sense of Definition~\ref{def:strictpos} if and only if it is faithful.
In practice, $\tau$ and $\tau_{+}$ will often be tracial states.

\subsection{Tomita-Takesaki Modular Theory}\label{sec:TomitaTakesaki}


Let $\A = B(\cH)$ be the algebra of bounded operators on a Hilbert space $\cH$. 
Let $\A_+ \subseteq \A$ be a factor, and let $\A_- = \A'_+$
be its commutant.
If $\Omega\in \cH$ is cyclic and separating for $\A_+$,  
then the modular involution $J \colon \cH \rightarrow \cH$
yields an antilinear homomorphism $\Theta(A) = JAJ$, with $\Theta(\A_+) = \A_-$.

This makes $\A$ a bosonic $q$-double of $\A_+$. Indeed, the algebras $\A_-$ and $\A_+$ 
commute by definition. To see that
the linear span of $\A_- \A_+$ is dense in $\A$, note that since $\A_+$ is a factor, 
$(\A_{-} \A_{+})' \subseteq \A_{+} \cap \A'_{+} = \C \one$.
Thus $(\A_{-} \A_{+})'' = \A$, and the linear span of $\A_-\A_+$
is dense in $\A$ by the double commutant theorem.

Although the state $\tau(A) = \lra{\Omega, A \Omega}$ is in general not factorizing, it is  
reflection positive on 
$\A_+$, since
\[
\tau(\Theta(A)A) = \lra{\Omega, JAJA\Omega} =  
\lra{JA^*\Omega, A\Omega} = 
\lra{\Delta^{1/2} A\Omega, A\Omega} \geq 0
\]
for all $A\in \A_+$.
Since the restriction $\tau_+$ of $\tau$ to $\A_+$ is faithful, it is strictly positive 
in the sense of Definition~\ref{def:strictpos}, with $\sharp$ given by the $*$-involution.

It was shown by Connes and Haagerup \cite{Connes1974,Haagerup1975}
that $\A_+$ is characterized up to unitary isomorphism by 
the \emph{natural positive cone} 
\[
 \mathcal{P}_{+}^{\natural} = \overline{ \{\Delta^{1/4} A^*A \,\Omega \,;\, A\in \A_+\}}\,,
\]
related to the reflection-positive cone $\cK_+$ of
Definition~\ref{defcone} by 
\[
\mathcal{P}_+^{\natural} = \overline{\cK_{+} \Omega}\,.
\]
From the fact that $\mathcal{P}_{+}^{\natural}$ is self-dual (as 
discovered independently by Araki \cite[Thm.~3]{Araki1974} and 
Connes \cite[Thm.~2.7]{Connes1974}), it follows that 
$\tau_H$ is reflection positive if and only if $e^{-H}\Omega$ is an element of 
$\mathcal{P}_{+}^{\natural}$.
In \S\ref{Sect:SforRP}, we provide 
tractable conditions on the Hamiltonian $H$ which ensure that this is the case.

\subsection{Grassmann Algebras}\label{ExIINumber1}
Classical fermions are described by
the Grassmann algebra
$\A = \bigwedge V$, where $V$ is 
an oriented 
Hilbert space of finite, even dimension $N$.
This is the unital algebra with generators $v, w \in V$ and relations 
\[
	v w + w v = 0\,.
\]
An arbitrary basis $\psi_i$ of $V$ yields generators satisfying
\beqs
	\psi_{i}\psi_j &=& - \psi_j \psi_i \quad\text{for}\quad i\neq j\\
	\psi^2_i &=& 0 \,\,\quad \qquad\text{for all} \quad i\,. 
\eeqs
The Grassmann algebra is $\Z$-graded, and hence in particular 
$\Z_2$-graded.
The degree of a homogeneous element 
$A \in \A$ is denoted by $\abs{A} \in \Z_{2}$, and we denote 
the even and odd part of $\A$ by $\A^0$ and $\A^1$.

Suppose that $V = V_{-}\oplus V_{+}$, where $V_{\pm}$ are 
Hilbert spaces of even dimension $n$.  It is important to keep track of the orientation of $V$, as this determines the sign of the Berezin integral.  
If $\psi_{1}, \ldots, \psi_{n}$ is a positively oriented orthonormal basis of $V_{+}$, 
then
\[\mu_{+} = \psi_{1} \wedge \ldots \wedge \psi_{n}\]
is a positively oriented volume on $V_{+}$. 
Similarly, if  $\mu_{-}$ is a positively oriented volume on $V_{-}$, then
the orientation of $V$ is defined by declaring 
$\mu=\mu_{-}\wedge\mu_{+} $ to be positively oriented. 
Since we are only working with vector spaces $V_{\pm}$ of even dimension, 
\be\label{eq:order}
\mu = \mu_{-} \wedge \mu_{+} 
=  \mu_{+} \wedge\mu_{-}\,.
\ee

%

The algebra $\A$ is the linear span of  $\A_{-}\A_{+}$, where
$\A_{\pm} = \bigwedge V_{\pm}$ are the Grassmann algebras of $V_{\pm}$.
Suppose that $\theta \colon V_{+} \rightarrow V_{-}$ is an antilinear, 
volume preserving 
vector space isomorphism. 
We extend it to an antilinear isomorphism $\theta \colon V \rightarrow V$ 
with $\theta^{2} = \mathrm{Id}$, and
$\theta(V_{+}) = V_{-}$.
Let $\Theta \colon \A \rightarrow \A$ be the unique antilinear 
homomorphism that agrees with $\theta$ on $V \subseteq \A$.
By \eqref{eq:order} and the fact that $\Theta(\mu_{\pm}) = \mu_{\mp} $, 
we find
\beq\label{eq:refvol}
	\Theta(\mu)
	= \Theta(\mu_{-}\wedge\mu_{+})
	= \mu_{+}\wedge\mu_{-}
	= \mu\,.
\eeq	
Note that for all $A_{\pm} \in \A$, we have
\[
	A_- A_+ = (-1)^{\abs{A_-}\abs{A_+}}A_+A_-\,.
\]
In particular, this holds for $A_\pm \in \A_\pm$, so $\A$ is the
fermionic  $q$-double 
of~$\A_{+}$. With $\zeta = i$, a functional $\varrho \colon \A \rightarrow \C$
is reflection positive if
\be\label{eq:rpmod2}
i^{\abs{A}^2}\varrho(\Theta(A)A) \geq 0 \quad\text{for}\quad A\in \A_+\,.
\ee


\begin{definition}
The {Berezin integral} is the functional $\tau \colon \A \rightarrow \C$ 
defined by
zero on $\bigwedge^{k}V$ for $k<\mathrm{dim}(V)$, and by $\tau(\mu) = 1$.
\end{definition}


\begin{proposition}\label{prop:BerezinFact}
If $\tau$ and $\tau_{+}$ are the Berezin integrals on $\A$ and $\A_{+}$, then $\tau$ is neutral, reflection invariant, and factorizing; 
$\tau(\Theta(A)\circ B) = \overline{\tau_{+}(A)} \tau_{+}(B)$
for all $A,B \in \A_{+}$ with $\abs{A} = \abs{B}$.
In particular, $\tau$ is reflection positive.
\end{proposition}
\begin{proof}
The Berezin integral is concentrated on 
$\bigwedge^{\mathrm{top}}V$, which is of degree $0 \in \Z_{2}$,
as $V$ is of even dimension.
Therefore, it is neutral. It is reflection invariant by 
\eqref{eq:refvol}, and reflection positivity follows from the 
factorization property by Proposition~\ref{invimppos}.

To prove factorization,
note that $\tau(\Theta(A)B) = \overline{\tau_{+}(A)} \tau_{+}(B)$ for all
${A,B \in \A_{+}}$.
 It therefore suffices to show that
$\tau(\Theta(A) \circ B) = \tau(\Theta(A)B)$. 
Here  we use that 
$\tau(\Theta(A)\circ B)$ is nonzero only if 
$A$ and $B$ are multiples of $\mu_+$. In that case, 
$\abs{A}=\abs{B} = 0$ in $\Z_{2}$, since $V_{+}$ is even dimensional, so that
$\Theta(A) \circ B = \Theta(A)B$. 
%
%
%
\end{proof}

For Grassmann algebras, the antilinear map
$\sharp \colon \A_{+} \rightarrow \A_{+}$ is  the \emph{Hodge star operator}.
It is defined by the requirement that 
\[A^{{\sharp}} \wedge B =  \lra{A,B} \,\mu_{+} \quad \text{for all}\quad A,B \in \A_{+}\,,\]
where $\mu_{+} = \psi_{1} \wedge \cdots \wedge \psi_{n}$ 
is the volume of an oriented basis of $V_{+}$.
The map $\sharp$ inverts the $\Z_{2}$-grading; since  
it maps  $\bigwedge^{k} V_{+}$ to $\bigwedge^{n-k} V_{+}$, 
and since $n = \mathrm{dim}(V_{+})$ is even, we find 
$\abs{A^{\sharp}} = \abs{A} = -\abs{A}$.
\begin{proposition}\label{Prop:berezinstrict}
The Berezin integral
$\tau_{+} \colon \A_{+} \rightarrow \C$ is strictly positive with respect to 
the hodge star operator
$\sharp$. 
\end{proposition}
\begin{proof}
We have
\[
\tau_{+}(A^{\sharp}A) = \lra{A,A}\,\tau_{+}(\mu_{+}) = \lra{A,A} > 0
\]
for all nonzero $A\in \A_+$.
\end{proof}
\begin{remark}
Note that unlike in \S\ref{sec:tensprod}, the functional  $\tau_{+}$ is 
\emph{not} the restriction of $\tau$ to $\A_{+} \subseteq \A$, since $\tau$ vanishes
identically on $\A_{+}$.
Also, in contrast to \S\ref{sec:TomitaTakesaki},
the antilinear map $\sharp$ is neither a homomorphism nor an anti-homomorphism. 
The Grassmann algebra is not a $*$-algebra, and the Berezin integral is not a state.
\end{remark}

\subsection{Clifford Algebras and CAR Algebras}\label{sec:Clifford}
A fermi\-onic quantum system is described by 
the Clifford algebra $\cA = \mathrm{Cl}(V)$.
Here $V$ is the complexification of a real Hilbert space $V_{\R}$
with inner product ${h \colon V_{\R} \times V_{\R} \rightarrow \R}$.
On the complex Hilbert space $V$, this gives rise to an inner product 
$\lra{v,w}$ and a bilinear form 
$h_{\C}(v,w)$.
 
The Clifford algebra $\cA = \mathrm{Cl}(V)$ is 
the unital algebra over $\C$ with generators $v \in V$ and relations 
\[
	v w + w v = 2h_{\C}(v,w)\one\,.
\]
Note that $\cA$ is $\Z_2$-graded. We denote
the degree of a homogeneous element $A \in \cA$ by $\abs{A} \in \Z_{2}$, and we denote 
the even and odd part of $\cA$ by $\cA^0$ and $\cA^1$, respectively.
The complex conjugation $v \mapsto \overline{v}$ on $V$
%
%
extends uniquely to an 
anti-linear anti-involution $A \mapsto A^*$ 
on $\mathrm{Cl}(V)$ that sends $\cA^0$ to $\cA^0$ and $\cA^1$ to $\cA^1$.
If $V_{\R}$ has an orthonormal
basis $\{c_{i}\}_{i\in S}$, 
then the operators $c_i \in \mathrm{Cl}(V)$ satisfy the 
\emph{Canonical Anticommutation Relations}
\begin{align}\tag{\mbox{CAR-1}}\label{eq:CR-1}
	c_{i}c_j &= -c_jc_i  \quad\text{for  } i\neq j\,, \\
	\tag{\mbox{CAR-2}}\label{eq:CR-2}
	c^2_{i} &= \one      \quad \qquad\text{for all  } i\,,\ \\
	\tag{\mbox{CAR-3}}\label{eq:CR-3}
	c^*_{i} &=  c^{-1}_{i}\,.
\end{align}

\begin{definition}\label{ref:deftracialcliff}
The \emph{tracial state} $\tau \colon \mathrm{Cl}(V) \rightarrow \C$
is the unique linear functional with $\tau(\one) = 1$ and 
$\tau(v_1 \cdots v_{k}) = 0$ if $v_1, \ldots, v_{k}$ are pairwise orthogonal
with respect to $h_{\C}$.
\end{definition}

Suppose that $V = V_+ \oplus V_-$ is an orthogonal decomposition
with respect to the bilinear form $h_{\C}$, 
and that $\theta \colon V \rightarrow V$ is an antilinear 
isomorphism of $(V,h_{\C})$ with $\theta(V_+) = V_-$ and $\theta^2 = \mathrm{Id}$.
Since $h_{\C}(V_+,V_-) = \{0\}$,
the subalgebras $\cA_{\pm} = \mathrm{Cl}(V_{\pm})$ 
\emph{supercommute};
\[
	A_-A_+ = (-1)^{\abs{A_+}\abs{A_-}} A_+A_-
\]
for all $A_{\pm} \in \cA_{\pm}$.
As $V = V_+ \oplus V_-$, the linear span of 
$\cA_-\cA_+$ is $\cA$, and 
$\cA_- \cap \cA_+ = \C \one$.

The map $\theta \colon V \rightarrow V$ 
extends uniquely to an antilinear homomorphism 
${\Theta \colon \cA \rightarrow \cA}$ with $\Theta(\cA_\pm) = \cA_{\mp}$.
It squares to the identity and inverts the grading, 
$\abs{\Theta(A)} = \abs{A} = -\abs{A}$.
It follows that $\cA$ is the fermionic $q$-double of $\cA_+$.
The reflection $\Theta$ is a $*$-homomorphism if and only if
$\theta(\overline{v}) = \overline{\theta(v)}$,
but we will not require this to be the case.

\begin{proposition}\label{prop:cliffordtau}
The tracial state $\tau$ on $\cA$ is neutral, 
faithful, factorizing, and reflection positive on $\cA_+$.
\end{proposition}
\begin{proof}
Neutrality is clear from the definition.
To see that $\tau$ is factorizing, i.e.\ $\tau(A_-A_+) = \tau(A_-)\tau(A_+)$ for $A_\pm \in \cA_\pm$, 
note that if 
both $A_- = v_1 \cdots v_k$ and 
$A_+ = w_1 \cdots w_l$ are products of pairwise orthogonal vectors,
then as $V_+ \perp V_-$, also $A_- A_+$ 
is a product of pairwise orthogonal vectors.
It follows that $\tau(A_-A_+)$ is zero for operators of this type.
Since every $A\in \cA$ can be decomposed as $A = \tau(A)\one + (A - \tau(A)\one)$
with the second term a sum of products of orthogonal vectors, 
we have $\tau(A_- A_+) = \tau(A_-)\tau(A_+)$ as required.
Reflection positivity therefore follows from Proposition~\ref{invimppos}
To show that $\tau$ is faithful, 
consider the linear map
$\iota \colon \bigwedge V \rightarrow \mathrm{Cl}(V)$ that sends 
$A = v_1 \wedge \ldots \wedge v_n$ to $\iota(A) = v_1 \cdots v_n$ if $v_1, \ldots, v_n$ are 
pairwise orthogonal. Since $\iota$ is a vector space isomorphism, and since
$\lra{A,A} = \tau(\iota(A)^*\iota(A))$, 
we have
$\tau(A^*A) > 0$ if ${A\neq 0}$.
\end{proof}

The CAR algebra $\A$ is the $C^*$-algebra defined as the norm closure 
of $\cA = \mathrm{Cl}(V)$ in $B(\cH_{\mathrm{GNS}})$.
Here $\cH_{\mathrm{GNS}}$ is the GNS-Hilbert space, the 
closure of $\cA$ with respect to the inner product 
$\lra{A,B}_{\tau} = \tau(A^*B)$.
Similarly, $\A_{\pm}$ is the norm closure of $\cA_{\pm}$.
As the linear span of $\A_- \A_+$ is norm dense in $\A$, 
the CAR algebra $\A$
is the $q$-double of $\A_+$.

\subsection{Parafermion Algebras and CPR Algebras}\label{sec:ParaCPR}

The motivation to consider reflection positivity for 
$\Z_p$-graded algebras comes from \emph{Para\-fermi\-on Algebras}.

Let $\Lambda$ be an ordered set, and let
$\vartheta \colon \Lambda \rightarrow \Lambda$ be an
order reversing, fixed point free involution.
Then $\Lambda$ is the disjoint union of $\Lambda_+$ and 
$\Lambda_- = \vartheta(\Lambda_+)$, 
where $\Lambda_+$ is the maximal subset with 
$\vartheta(\lambda) < \lambda$ for all $\lambda \in \Lambda_+$.
A typical example is 
$\Lambda = \Z + 1/2$, with $\vartheta(\lambda) = -\lambda$
and $\Lambda_{\pm} = \pm(\N + \frac{1}{2})$.


A collection of para\-fermions of order $p$ is a family of operators  
$c_{\lambda}$, indexed by $\Lambda$.  
The parafermions are characterized by a 
primitive $p^{\rm th}$ root $q$ of unity, and satisfy
the \emph{Canonical Parafermion Relations}
	\begin{align}\tag{\mbox{CPR-1}}\label{eq:PF-1}
		c_{\lambda}c_{\lambda'} &= \q\, c_{\lambda'}c_{\lambda}\;,  \qsp{for} 
		\lambda<\lambda'\;, \\
		\tag{\mbox{CPR-2}}\label{eq:PF-2}
		c_{\lambda}^{p} & =  1\;,\\
		c_{\lambda}^{*} &= c_{\lambda}^{-1} \;.
		\tag{\mbox{CPR-3}}\label{eq:PF-3}
	\end{align}
If we set the degree of each $c_{i}$ equal to 1, then
the algebra generated by the parafermions is graded by 
$\Z_{p} = \Z/p\Z$.
 
\begin{definition}
The Parafermion Algebra $\cA(q,\Lambda)$
is the  $\Z_{p}$-graded $*$-algebra generated by the parafermions 
$c_{\lambda}$ of degree $p$.
\end{definition}	
 
Denote the degree of $A \in \cA(q,\Lambda)$ by $\abs{A} \in \Z_{p}$.
Products of para\-fermions provide a natural basis $\{C_{I}\}_{I\in \mathcal{I}}$ 
for $\cA(q,\Lambda)$.  
The elements are labelled by the set  $\mathcal{I} = \Z_{p}^{(\Lambda)}$  of maps
$I \colon \Lambda \rightarrow \Z_p$ with $I_{\lambda} \neq 0$ 
for only finitely many $\lambda \in \Lambda$.
For $I \in \Z_{p}^{(\Lambda)}$, define the basis element
\be\label{BasisPF}
C_{I} =\overrightarrow{ \prod_{\lambda\in \Lambda}}c_\lambda^{\,I_{\lambda}}\,,
\ee
where $\overrightarrow{\prod}_{\lambda\in \Lambda}$ indicates that the order of the factors 
$c_\lambda^{I_{\lambda}}$ in the product respects the order of $\Lambda$.   
Note that $C_{0} = \one$ is the identity in $\cA(q,\Lambda)$.

\begin{definition}
The \emph{tracial state} $\tau \colon \cA(q,\Lambda) \rightarrow \C$ 
is the linear functional with $\tau(\one) = 1$ and $\tau(C_{I}) = 0$ for $I\neq 0$.
\end{definition}

\begin{proposition}\label{prop:prestateCPR}
This is indeed a faithful, tracial state on $\cA(q,\Lambda)$.
\end{proposition}
\begin{proof}
It suffices to show that 
\be
	\tau(C^*_{I}C_J) = \tau(C_{J}C^*_{I}) = \delta_{IJ}\,.
\ee
For this, note that the basis $C_{I}$ transforms under the anti-linear anti-involution $*$  as
\be\label{eq:badorder}
	C_{I}^* = \overleftarrow{\prod_{\lambda \in \Lambda}}c_\lambda^{-I_{\lambda}}\,,
\ee
where   $\overleftarrow{\prod}_{i\in S}$ denotes that the order of the factors 
is the reverse to their order in $\Lambda$. 
Since $C^*_{I}C_{J}$ and $C_{J}C^*_{I}$ are multiples of $C_{J-I}$, the expression 
$\tau(C^*_{I}C_J)$
is zero if $I \neq J$. If $I=J$, then $\tau(C^*_{I}C_{I}) = \tau(C_{I}C^*_{I}) =1$.
\end{proof}

Let $\cH_{\mathrm{GNS}}$ be the Hilbert space closure of 
$\cA(q,\Lambda)$ with respect to the nondegenerate inner product 
$\lra{A,B}_{\tau} = \tau(A^*B)$.
It carries the usual GNS-representation of $\cA(q,\Lambda)$, and has 
orthonormal Hilbert basis $\{C_{I}\}_{I \in \Z_{p}^{(\Lambda)}}$. 

\begin{definition}\label{def:CPR}
The CPR algebra $\A(q,\Lambda)$ is the $\Z_p$-graded $C^*$-algebra 
arising as the norm closure of $\cA(q,\Lambda)$ 
in $B(\cH_{\mathrm{GNS}})$.
\end{definition} 

Define the antilinear isomorphism
$\Theta \colon \cA(q,\Lambda) \rightarrow \cA(q,\Lambda)$
of the para\-fermion algebra 
by
\[
	\Theta(c_{\lambda}) = c^{-1}_{\theta(\lambda)}.
\]

\begin{proposition}
The tracial state
$\tau$ on $\cA(q,\Lambda)$ is reflection invariant; 
for all $A \in \cA(q,\Lambda)$, we have
\be\label{eq:RIpara}
	\tau(\Theta(A)) = \overline{\tau(A)}\,.
\ee 
\end{proposition}
\begin{proof}
Denote by $\theta(I)$ the index with $\theta(I)_{\lambda} = I_{\vartheta(\lambda)}$. 
Then $\Theta(C_I)$ is a multiple of $C_{-\theta(I)}$, and
$\Theta(C_{0}) = \one = C_0$. Since $\tau(C_{I}) \tau(C_{-\theta(I)})= 0$ for 
$I \neq 0$ and $\tau(C_0) = 1$,
we have $\tau(\Theta(C_I)) = \tau(C_I)$ for all $I \in \Z^{(\Lambda)}_p$.
As $\Theta$ is antilinear, this shows that $\tau$ is reflection invariant.
\end{proof}

It follows that $\Theta$ extends to an antilinear homomorphism of 
the $C^*$-algebra $\A(q,\Lambda)$, and that \eqref{eq:RIpara}
holds for all $A\in \A(q,\Lambda)$.


Let $\A_{\pm}(q,\Lambda)$ be the norm closure of the algebra generated by 
the parafermions $c_{\lambda}$ with $\lambda \in \Lambda_{\pm}$.
Then $\A(q,\Lambda)$ is the $q$-double of $\A_+(q,\Lambda)$.
Indeed,
$\A(q,\Lambda)$ is the norm closure of the linear span of the product 
$\A_{-}(q,\Lambda)\A_{+}(q,\Lambda)$.
Furthermore, it follows from \eqref{eq:PF-1} that 
homogeneous elements $A_{\pm} \in \A_{\pm}(q,\Lambda)$
satisfy
\[
A_- A_+ = q^{\abs{A_-}\abs{A_+}} A_+A_-\,.
\]


\begin{proposition}\label{prop:paratraceisnice}
The tracial state $\tau$ extends to a neutral, 
faithful, factorizing, reflection invariant state on $\A(q,\Lambda)$, which is
reflection positive on $\A_{+}(q,\Lambda)$.
\end{proposition}
\begin{proof}
The state $\tau$ extends from $\cA(q,\Lambda)$ to $\A(q,\Lambda)$
by $\tau(A) = \lra{\Omega, A\Omega}$, where $\Omega$ is the 
cyclic vector in $\cH_{\mathrm{GNS}}$.
Reflection invariance follows from the previous discussion, and it is clear from 
the definition that $\tau(A) = 0$ if $A$ is homogeneous 
of nonzero degree. 
It is faithful since $\tau(C^*_{I}C_J) = \delta_{IJ}$, cf.\ Proposition~\ref{prop:prestateCPR}.

If $I_{\pm} \in \Z_p^{\Lambda_{\pm}}$, then $C_{I_-}C_{I_+}$ 
is proportional to $C_{I_- + I_+}$. It follows that 
$\tau(C_{I_-} C_{I_+})$ is zero unless $I_{-} = I_+ = 0$, in which case it 
is equal to~$1$.
If we expand $A_{\pm} \in \A_{\pm}(q,\Lambda)$ in a norm convergent sum
\[
	A_{\pm} = \sum_{I \in \Z_p^{\Lambda_{\pm}}} a^{\pm}_{I} C_I,
\]
then $\tau(A_- A_+) = a^-_{0}a^+_{0}$. Since $\tau(A_{\pm}) = a^{\pm}_{0}$, 
it follows that $\tau$ factorizes, $\tau(A_-A_+) = \tau(A_-) \tau(A_+)$.
Reflection positivity then follows from Proposition~\ref{invimppos}.
\end{proof}

\begin{remark}
Note that in the fermionic case $q=-1$, the relations CPR 1--3 
for the CPR algebra coincide with the relations 
CAR 1--3 for the CAR algebra.
Thus
the CPR algebra for $p=2$ and $q=-1$ is isomorphic to the CAR algebra.
 \end{remark}

\setcounter{equation}{0}
\section{Sufficient Conditions for Reflection Positivity}\label{Sect:SforRP}

We return to the general setting of \S\ref{sec:section2},
which in particular encompasses the applications in the previous section. We assume that:
\begin{itemize}
\item[Q1.] \quad $\A$ is a $\Z_{p}$-graded, locally-convex, topological algebra.
\item[Q2.] \quad $\A$ is the $q$-double of $\A_{+}$.
\item[Q3.] \quad $\tau \colon \A \rightarrow \C$ is continuous, neutral and reflection positive.
\end{itemize}
In this setting, and under the natural condition that $\exp \colon \A \rightarrow \A$
is continuous,
we give the following
\emph{sufficient} condition on an element $H\in \A^0$ of degree zero for its
Boltzmann functional $\tau_{H}(A) = \tau(Ae^{-H})$ to be reflection positive.
Namely, 
this is the case if $H$ admits a decomposition 
\be\label{eq:decompositionBoven}
H = H_{-} + H_{0} + H_{+}\,,
\ee
with $H_{+}\in \A_{+}$, with $-H_{0} \in \cok$,
and with $H_{-} =\Theta(H_{+})$.

We will show in
\S\ref{Sect:N-SforRP} that these conditions are \emph{necessary} as well as sufficient,
under the additional assumptions that $\tau$ 
is factorizing and strictly positive.
These extra assumptions are \emph{not} needed 
for the results in the present section. This is relevant for applications in
quantum field theory and Tomita-Takesaki theory, where the
background functional $\tau$ is generally not factorizing.


\subsection{The Reflection-Positive Cone} \label{Sect:RPCone}
The results in this section rely heavily on the reflection positive cone $\cK_+$, 
introduced in \S\ref{sec:firstcone}.
By Proposition~\ref{prop:rpclosure}, the (continuous) dual cone of $\cK_+$ is the set of
(continuous), neutral, reflection positive functionals $\varrho \colon \A \rightarrow \C$.
A key point in the characterization of reflection positive functionals 
is proving that $\cK_+$ is multiplicatively closed.

%
%

\begin{theorem}\label{lem:KeyLemma}
The cone $\cK_{+}$ is multiplicatively closed, and it is pointwise invariant under
reflection. Namely, $\cK_{+} \cK_{+} \subseteq \cK_{+}$, and $\Theta|_{\cK_{+}} = \mathrm{Id}$.
\end{theorem}

The proof uses the following two lemmas. We formulate them separately, 
as we will need them later on.

\begin{lemma}\label{Prop:TProduct-RI}
Let $A,B\in \A_{+}$. 
Then the reflection of $\Theta(A)\circ B$ is ${\Theta(B)\circ A}$.
In particular, $\Theta(A)\circ A$ is reflection invariant. 
\end{lemma}

\begin{proof}
For $A,B\in\A_{+}$ of homogeneous degree $\abs{A}=\abs{B}$, one has
	\beqs
		\Theta(\Theta(A)\circ B)
		&=& \Theta\lrp{\zeta^{\abs{A}\abs{B}}\,\Theta(A)B}
		= \overline{\zeta}^{\abs{A}\abs{B}}\, A\Theta(B)\nn
		&=& \overline{\zeta}^{\abs{A}\abs{B}}\,
		q^{\abs{A}\abs{B}}\, \Theta(B) A
		=\Theta(B)\circ A\;,
	\eeqs
as claimed. If $\abs{A}\neq\abs{B}$, then the twisted product is zero.
\end{proof}

\begin{lemma}\label{Prop:TProduct-MClosed}
If $A_{1}, A_{2}, B_{1}, B_{2}$ in $\A_{+}$ are homogeneous 
with $\abs{A_{1}} = \abs{B_{1}}$ and $\abs{A_{2}} = \abs{B_{2}}$,
then 
	\be
		(\Theta(A_{1})\circ B_{1}) \ (\Theta(A_{2})\circ B_{2})
		= \Theta(A_{1}A_{2}) \circ B_{1}B_{2}\;.
	\ee 
\end{lemma}

\begin{proof}
Note that 
\beqs
(\Theta(A_{1})\circ B_{1}) \ (\Theta(A_{2})\circ B_{2}) \!&=&\!
\zeta^{\abs{A_{1}}^2} \zeta^{\abs{A_{2}}^2}\,\Theta(A_{1})\,B_{1}\,
\Theta(A_{2})\,B_{2}\\
&=&\!
\zeta^{\abs{A_{1}}^2} \zeta^{\abs{A_{2}}^2}
q^{-\abs{B_{1}}\abs{\Theta(A_{2})}} \Theta(A_{1})\Theta(A_{2})\,B_{1}B_{2}\\
&=&\!
\zeta^{(\abs{A_{1}}+\abs{A_{2}})^2}\,\Theta(A_{1}A_{2})\,B_{1}B_{2}\,. 
\eeqs
Here we use $\abs{B_{1}} = \abs{A_{1}}$, $\abs{\Theta(A_{2})}=-\abs{A_{2}}$ and $q=\zeta^{2}$. As $\abs{A_{1}A_{2}}= \abs{A_{1}}+\abs{A_{2}}$, the final expression equals $\Theta(A_{1}A_{2})\circ B_{1}B_{2}$. 
\end{proof}

\begin{proof} [Proof of Theorem \ref{lem:KeyLemma}]
The cone $\cK_+$ is
reflection invariant by Lem\-ma~\ref{Prop:TProduct-RI}, 
and multiplicatively closed by Lemma~\ref{Prop:TProduct-MClosed}.
\end{proof}

As we will mainly be interested in \emph{continuous} reflection positive 
functionals, we will need to extend Theorem~\ref{lem:KeyLemma}
to the closure $\cok$ of the convex hull $\coh$ of $\cK_+$.
Note that $\cok\subseteq \A^{0}$, as
every element of $\cK_{+}$ is of degree zero.
By polarization, we obtain a useful characterization of~$\coh$.


\begin{proposition}\label{prop:charpos}
Let $K\in \A^{0}$. Then $K\in{\coh} $
if and only if~both:
\begin{enumerate}
\item
The element $K$ can be written as a finite sum 
\be\label{eq:posdefstandard}
K = \sum_{I,J \in \mathcal{I}}J_{IJ}\,\Theta(C_I)\circ C_J\,,
\ee
with $C_{I}\in \A_{+}$ labelled by a finite set $\mathcal{I}$.
\item
Let $(J_{IJ})_{\mathcal{I}}$ be the matrix with entries $J_{IJ}$, labelled
by $I,J \in\cI$. Then
$(J_{IJ})_{\mathcal{I}}$ is positive semi-definite, and
$J_{IJ} = 0$ if ${\abs{C_I} \neq \abs{C_J}}$.
\end{enumerate}
\end{proposition}

%
%
%
\begin{proof}
Every $K \in {\coh} $ can be written as a finite convex combination $K = \sum_{r\in \mathcal{I}} p_{r} \Theta(X_r)\circ X_r$ of elements of $\cK_{+}$.
In particular, it is of the form \eqref{eq:posdefstandard} with $J_{IJ} = p_{I}\delta_{IJ}$.
Conversely, suppose that $K$ is of the form \eqref{eq:posdefstandard}.
Since the matrix $(J_{IJ})_{\mathcal{I}}$ is positive semidefinite, 
its eigenvalues $p_{r}$ are nonnegative.
As $(J_{IJ})_{\mathcal{I}}$ respects the grading, the corresponding eigenvectors $(x^{r}_{I})_{I\in \mathcal{I}}$
can be chosen so that 
$x^{r}_{I} =0$ unless $C_{I}$ has a fixed degree $\abs{C_{I}} = d_{r}$.
It follows that $X_{r} = \sum_{I}x^{(r)}_{I} C_{I}$ is
homogeneous of degree $\abs{X_r} = d_{r}$, and
$K = \sum_{r} \Theta(X_{r}) \circ X_{r}\in {\coh} $.
\end{proof}
%
We can thus characterize the closure $\cok$ of ${\coh} $ as follows:
\begin{corollary}
An element $K\in \A^{0}$ is in 
$\cok$ if and only if
$K = \lim_{n\rightarrow \infty} K_{n}$,
with $K_{n} \in \cK_{+}$ as in \eqref{eq:posdefstandard}.
\end{corollary}

\begin{corollary}\label{cor:keylemma}
The closed, convex cone $\cok$ is multiplicatively closed, and it is pointwise invariant under
reflection,
\[\cok \cdot \cok\subseteq \cok\,,
\quad\text{and}\quad
\Theta|_{\cok} = \mathrm{Id}
.\]
\end{corollary}

%

\begin{proof}
As $\Theta \colon \A \rightarrow \A$ is a continuous $\R$-linear map, $\Theta|_{\cK_{+}} = \mathrm{Id}$
implies that $\Theta|_{\cok} = \mathrm{Id}$.
To prove that $\cok\cdot\cok\subseteq \cok$, note that
as multiplication ${\A\times\A} \rightarrow \A$ is separately continuous,
left multiplication by $A\in \cK_{+}$ is a continuous linear map 
$L_{A} \colon \A \rightarrow \A$.
As $L_{A}(\cK_{+}) \subseteq \cK_{+}$,
we have $L_A(\cok) \subseteq \cok$
for every $A\in \cK_{+}$. It follows that
$\cK_{+} \cdot  \cok\subseteq \cok$. 
In particular,
right multiplication 
$R_{B} \colon \A \rightarrow \A$ by $B\in \cok$
satisfies $R_{B}(\cK_{+}) \subseteq \cok$.
As $R_{B}$ is a continuous linear map, it follows that
$R_{B}(\cok) \subseteq \cok$.
%
We conclude that
$\cok \cdot \cok \subseteq \cok$, as desired.
\end{proof}

\subsection{Sufficient Conditions for RP}

Using the fact that $\cok$ is multiplicatively closed, we obtain the following 
criterion for reflection positivity.

\begin{proposition}\label{prop:conetostate}
Let $A \mapsto \tau(A)$ be a continuous, reflection positive functional on $\A$,
and let $K_{1}, K_{2}\in\cok$.
Then the functionals 
\[
A \mapsto  \tau(K_{1}A),\quad 
A \mapsto  \tau(AK_{2}), \quad\text{and} \quad A \mapsto \tau(K_{1}AK_{2})
\] 
are 
also continuous and reflection positive.
\end{proposition}
\begin{proof}
In light of Proposition~\ref{prop:rpclosure}, it suffices to prove that 
if $A\in \cok$, then also $KA$, $AK$, and $KAK$ are in 
$\cok$.  
This follows from 
Corollary~\ref{cor:keylemma}.
As multiplication is separately continuous, continuity of the 
above three functionals follows from continuity of $\tau$.
\end{proof}

\begin{proposition}\label{lem:ExpposinK}
Suppose that $-H \in \cok$, and that the exponential series
\be\label{eq:expseries}
\exp(- H)  - I = \sum_{k=1}^{\infty} \frac{1}{k!} (-H)^{k}
\ee
converges in $\A$.
If $\tau$ is a continuous, reflection positive functional, then also 
the Boltzmann functional 
\[\tau_{H} (A) = \tau(A\,e^{-H} )\] is continuous and reflection positive. 
Its reflection positive inner product dominates that of $\tau$,
\be\label{eq:dominates}
\lra{A,A}_{\Theta, \tau_{H}} \geq \lra{A,A}_{\Theta,\tau} 
\quad\text{for all}\quad A \in \A_{+}\,.
\ee
\end{proposition}

\begin{proof}
As $-H \in \cok$, every term $\frac{1}{k!} (-H)^{k}$
is in $\cok$ by Corollary \ref{cor:keylemma}.
Since $\cok$ is a convex cone, the same holds for the partial sums
in equation \eqref{eq:expseries}, and as
$\cok$ is closed, also the limit
$K_{2} := e^{-H} - I$ is in $\cok$.
If $\tau$ is continuous and reflection positive, then
by Proposition~\ref{prop:conetostate}, the functional
$A\mapsto \tau(A(e^{-H}-I) )$ is also continuous and reflection positive.
It follows that 
	\[
	\tau_{H}(\Theta(A)\circ A)
	=\tau((\Theta(A)\circ A)e^{-H})
	\geq \tau(\Theta(A)\circ A)\geq 0\,,
\]
for all $A\in \A_{+}$.
In particular $\tau_{H}$ is reflection positive. 
\end{proof}

\begin{remark}
We study the 
reflection-positivity properties of the functional   
$\tau_{H}(A)=\tau(A\,e^{-H})$ in some detail.
Using Proposition \ref{prop:conetostate}, 
one sees that
similar results hold for 
the functionals 
\[ {}_{H}\tau (A) = \tau(e^{-\beta H}\,A) \quad \text{and} \quad 
{}_{H_{1}}\tau_{H_{2}}(A) = \tau(e^{-\beta H_{1}}\,A\,e^{-\beta H_{2}})\,.
\]
%
%
\end{remark}

\begin{theorem}[\bf Sufficient Conditions for RP]\label{Thm:sufficient}
Suppose that the exponential series 
$\exp(A) = \sum_{k=0}^{\infty} \frac{1}{k!}A^{k}$
converges for all $A\in \A$, and that 
$\exp \colon \A \rightarrow \A$ is continuous.
Let $\tau \colon \A \rightarrow \C$ be a continuous, neutral, reflection positive functional. 
Let $H\in \A$ have degree zero, and  
admit a decomposition 
\be\label{eq:decomposition}
H = H_{-} + H_{0} + H_{+}\,,
\ee
with $H_{+}\in \A_{+}$, with $-H_{0} \in \cok$,
and with $H_{-} =\Theta(H_{+})$. Then the
Boltzmann functional $\tau_{H}(A) = \tau(A\,e^{-H})$
is continuous and reflection positive.
%
\end{theorem}
\begin{proof}
For $\varepsilon >0$, define $H_{\varepsilon} \in \A$ by
\[
	H_{\varepsilon} = 
	H_{0} - \Theta(\varepsilon^{-1} I - \varepsilon H_{+}) (\varepsilon^{-1} I - \varepsilon H_{+})\,.
\] 
Let $A\in \A_+$ be homogeneous.
As $-H_{\varepsilon} \in \cok$,
Proposition~\ref{lem:ExpposinK} yields 
\[
	\tau((\Theta(A)\circ A)\,e^{-H_{\varepsilon}}) \geq 0\,.
\] 
Note that
$H_{\varepsilon} = H - \varepsilon^{-2}I - \varepsilon^2\Theta(H_+)H_+$.
By Remark \ref{Remark:Constant}, the additive constant $\varepsilon^{-2}I$ does not change reflection positivity. Therefore, $H'_{\varepsilon} = H  - \varepsilon^2\Theta(H_+)H_+$ satisfies
\[
	\tau((\Theta(A)\circ A) \,e^{-H'_{\varepsilon}})
	\geq 0 \,.
\] 
Since $\lim_{\varepsilon \downarrow 0} H_{\varepsilon}'  = H$ 
and $\exp\colon \A \rightarrow \A$
is continuous, this yields 
\[
	\lim_{\ep\downarrow0}
	\tau((\Theta(A)\circ A) \,e^{-H'_{\varepsilon}})
	=
	\tau((\Theta(A)\circ A)\,e^{- H})
	\geq 0 \,,
\] 
as required.
\end{proof}

\begin{remark}\label{rk:beta}
In fact, the Boltzmann functional $\tau_{\beta H}$ is reflection positive for all $\beta \geq 0$
if $H$ satisfies the conditions of Theorem~\ref{Thm:sufficient}.
\end{remark}

\setcounter{equation}{0}
\section{Necessary and Sufficient Conditions for RP}\label{Sect:N-SforRP}
In order to obtain necessary as well as sufficient conditions for reflection 
positivity, we now introduce a more rigid framework. 
Let $\A$ be the $q$-double of $\A_+$, and let $\tau$ be a neutral, reflection 
positive functional on~$\A$.
In addition to the previous assumptions Q1--3,
we now require the following, additional properties, described in more detail in
\S\ref{sec:factorization} and \S\ref{sec:sharpsec}, and in \S\ref{sec:HomOrthBas} below.
\begin{itemize}
\item[Q4.\,\,]
The algebra $\A_+$ comes with 
an antilinear, grading-inverting map $\sharp:\A_{+}\to\A_{+}$. 
We require that $\A_+$ admits an unconditional, homogeneous Schauder basis.
\item[Q5.\,\,] The functional $\tau\colon \A \rightarrow \C$ \emph{factorizes} into $\tau_+$ and $\tau_-$. 
\item[Q6.\,\,] The functional $\tau_+ \colon \A_+ \rightarrow \C$ is \emph{strictly positive} 
for~$\sharp$. 
%
\end{itemize}

These additional assumptions  are suitable in the context of statistical physics, where one has a uniform background measure or tracial state.  This state is generally assumed to be 
faithful, reflection invariant, and 
factorizing.

In quantum field theory however, it is necessary to put nearest-neighbor 
 couplings into the background measure, in order to define it mathematically.
This destroys the factorization property; in the case of quantum fields, the 
results on sufficient conditions
in the previous section
still apply, while the results on necessary conditions 
in the present section need to be strengthened.

\subsection{The Matrix of Coupling Constants}\label{sec:HomOrthBas}
Let $\tau$ be a
factorizing functional such that $\tau_{+}$ is strictly positive.
Then the scalar product $\lra{A,B} = \tau_{+}(A^{\sharp}B)$ on $\A_{+}$ is nondegenerate.  

We require that $\A_+$ has a countable, homogeneous Schauder basis. 
This is a countable, ordered set $\{v_{I}\}_{I\in \mathcal{I}}$ of homogeneous elements 
such that every $A\in \A_+$ has a 
unique expansion $A = \sum_{I\in \mathcal{I}} a_I v_I$.
Using the Gram-Schmidt procedure, one can find a 
homogeneous Schauder basis
$\{C_{I}\}_{I\in \mathcal{I}}$ of $\A_+$ with the following properties:
\begin{itemize}
\item[B1.\,\,] There is a unit $C_{I_0} = \one$, for some distinguished index $I_0 \in \mathcal{I}$.
\item[B2.\,\,] For all $I,J \in \mathcal{I}$, one has $\tau_{+}(C_I^{\sharp} C_J) = \delta_{IJ}$.
\item[B3.\,\,] The linear span of $\{C_{I}\}_{I\in \mathcal{I}}$ is dense in $\A_+$.
\end{itemize}
Note that any set $\{C_{I}\}_{I\in \mathcal{I}}$ of homogeneous elements 
satisfying B1--3 is a Schauder basis;
every $A \in \A_+$ has a unique expansion $A = \sum_{I\in \mathcal{I}} a_I C_I$,
with $a_I = \tau(C_{I}^{\sharp}A)$.  

We use the orthogonal basis of $\A_+$ to construct a basis of 
$\A^0$.
If $\abs{C_I} = \abs{C_J}$, define the operators $B_{IJ}, \widehat{B}_{IJ} \in \A^{0}$ by
	\be\label{eq:adaptedbasis}
		B_{IJ} 
		=  \Theta(C_{I})\circ C_{J}\;,
		\quad\text{and} \quad
		\widehat{B}_{IJ}
		= \Theta(C_{I}^{\sharp})\circ C_{J}^{\sharp}\;.
	\ee
\begin{lemma}\label{Lemma:DualBasis}
The operators $B_{IJ}$ and $\widehat{B}_{IJ}$, for  $\abs{C_{I}}=\abs{C_{J}}$, are dual in the sense that  
\be\label{DualityIdentity}
\tau(\widehat{B}_{IJ}\,B_{I'J'} ) = \delta_{II'}\delta_{JJ'}\,.
\ee
\end{lemma}

\begin{proof}
Using Lemma~\ref{Prop:TProduct-MClosed}, the factorization property of 
$\tau$, and the fact that $\sharp$ inverts the grading,
one finds
\beq
\tau(\widehat{B}_{IJ}\,B_{I'J'}) &=& \tau(\Theta(C^{\sharp}_{I} C_{I'})\circ 
C^{\sharp}_{J}C_{J'}) \label{eq:expprod} \\
&=& 
\overline{\tau_{+}(C_{I}^{\sharp} \,C_{I'}) }\,
\tau_{+}(C_{J}^{\sharp}\,C_{J'})\;. \nonumber
\eeq
The lemma follows since 
$\tau_{+}(C_{I}^{\sharp} \,C_{I'}) = \delta_{II'}$, and
$\tau_{+}(C_{J}^{\sharp}\,C_{J'}) = \delta_{JJ'}$.
%
\end{proof}

As the linear span of $\A_- \A_+$  is dense in $\A$, every $A\in \A^0$ has a convergent expansion
\be\label{eq:basexpansion}
A = \sum_{(I,J)\in \cI \times \cI} a_{IJ} B_{IJ}\,,
\ee
with $a_{IJ} = 0$ if $\abs{C_I} \neq \abs{C_J}$. The sum requires an order on 
$\cI \times \cI$, which is obtained in a natural way from the order on $\cI$.

\begin{proposition}\label{Prop:UniqueExpansion}
The expansion \eqref{eq:basexpansion} of $A \in \A^0$ is unique, 
and the coefficients $a_{IJ} = \tau(\widehat{B}_{IJ}A)$ depend continuously on $A$.
\end{proposition}

\begin{proof}
The uniqueness and continuity of the coefficients $a_{IJ}$ follows from
the explicit expression, which is a consequence of Lemma~\ref{Lemma:DualBasis}.  
\end{proof}

\begin{proposition}\label{prop:refinvmatrix}
The operators $B_{IJ}$ satisfy $\Theta(B_{IJ}) = B_{JI}$.
Therefore, $A \in \A^{0}$ is reflection invariant, $\Theta(A) = A$, if and only if the matrix 
$(a_{IJ})_{\mathcal{I}}$ is Hermitian, $a_{JI} = \overline{a_{IJ}}$.
\end{proposition}

\begin{proof}
The second statement follows from the first by uniqueness of the expansion 
\eqref{eq:basexpansion}.  The property $\Theta(B_{IJ}) = B_{JI}$ follows 
immediately from Lemma~\ref{Prop:TProduct-RI}.
\end{proof}

In particular, every Hamiltonian $H\in \A^{0}$ of degree zero has an  unique expansion
\be\label{eq:Hexpansion}
-H = \sum_{(I,J) \in \cI \times \cI} J_{IJ}\, \Theta(C_{I})\circ C_{J}\,,
\ee
with $J_{IJ} = 0$ unless $\abs{C_I} = \abs{C_J}$.
The matrix $(J_{IJ})_{\mathcal{I}}$ describes the \emph{couplings} between
$C_{J}\in \A_{+}$ and $\Theta(C_I) \in \A_-$.
\begin{definition}[\bf The Matrix of Coupling Constants]
The matrix $(J_{IJ})_{\mathcal{I}}$ is called the matrix of coupling constants.
\end{definition}

The term $J_{I_{0} I_{0}}$ in the coupling matrix describes the coefficient of the identity, 
an irrelevant additive constant in $H$.  
Since $C_{I_0} = \one$, the terms
$J_{I_{0} J}$ describe couplings inside $\A_+$.
Similarly, the terms $J_{I I_{0}}$ describe couplings inside~$\A_-$. 
Finally, the terms $J_{IJ}$ with $I\neq I_{0}$ and $J \neq I_{0}$ describe couplings 
between $\A_-$ and $\A_+$.

\begin{definition}[\bf Couplings Across the Reflection Plane]
The submatrix $(J^{0}_{IJ})_{\mathcal{I}\backslash \{I_{0}\}}$ of $(J_{IJ})_{\mathcal{I}}$,
consisting of elements with $I, J \neq I_{0}$, is called the matrix of coupling constants across the 
reflection plane. 
\end{definition}

\begin{proposition}\label{prop:couplingposdef}
If the matrix of coupling constants $(J_{IJ})_{\mathcal{I}}$ is Hermitian, then
$H$ is reflection invariant. If it is 
positive semidefinite, then $-H \in \cok$.
\end{proposition}
\begin{proof}
The first statement follows from Proposition~\ref{prop:refinvmatrix}.
The second follows from
Proposition~\ref{prop:charpos}, since
every finite partial sum of \eqref{eq:Hexpansion} 
is of the form \eqref{eq:posdefstandard}
if the matrix $(J_{IJ})_{\mathcal{I}}$ is positive semidefinite.
\end{proof}

\begin{remark}
In applications, the operator $H \in \A^0$ is often given in terms of a coupling matrix 
$(J_{IJ})_{\cI}$, by way of the expansion \eqref{eq:Hexpansion}.
Combining Proposition \ref{prop:couplingposdef} with
Theorem \ref{Thm:sufficient}, we see that $\tau_{H}$
is reflection positive if 
$(J_{IJ})_{\mathcal{I}}$ is Hermitian, with 
positive semidefinite submatrix 
$(J^{0}_{IJ})_{\mathcal{I}\backslash \{I_{0}\}}$ of couplings across the reflection plane.
These properties are easy to check in concrete situations.
\end{remark}

\subsection{Necessary Conditions for RP}
In this section, we prove necessary 
conditions on the matrix of coupling constants across the reflection plane for the Boltzmann functional $\tau_{\beta H}$ to be reflection positive.
In \S\ref{sec:Characterization}, we will show that these are equivalent to the 
sufficient conditions in \S\ref{Sect:SforRP}.

\begin{lemma}\label{lemma:nobasis}
Suppose that
the exponential series for $\exp(-\beta H)$ converges, 
and is differentiable at $\beta = 0$.
If $\tau_{\beta H}$ is reflection positive for all $\beta \in [0,\varepsilon)$,
then
\be\label{eq:necessary}
\tau((\Theta(A)\circ A)H)\leq 0\,,
\ee
for all $A \in \A_{+}$ with
$\tau\lrp{\Theta(A)\circ A}=0$.
\end{lemma}
\begin{proof}
Consider the function 
$
F(\beta) 
=
\tau((\Theta(A)\circ A)\,e^{-\beta H})\geqslant 0
$.
At $\beta = 0$, one finds
$F(0) =\tau(\Theta(A)\circ A) =  0$.
Hence
\[
\left. -\frac{d}{d\beta}F(\beta)\right|_{\beta = 0} 
 =  
\tau((\Theta(A)\circ A)\,H)
=\lim_{\beta\downarrow0} -\frac{F(\beta)}{\beta}\leqslant0\,,
\]
as claimed.
\end{proof}

\begin{theorem}\label{thm:necessary}
Suppose that there exists an $\varepsilon > 0$ such that 
the map
$\beta \mapsto \exp(-\beta H)$ is well defined on $\beta \in [0,\varepsilon)$,
and differentiable at $\beta = 0$.
If $\tau_{\beta H}$ is reflection positive for all $\beta \in [0,\varepsilon)$,
then the matrix $(J^{0}_{IJ})_{\mathcal{I}\backslash \{I_{0}\}}$ of coupling constants across the reflection plane is positive semidefinite.
\end{theorem}
\begin{proof}
Let $A\in\A_{+}$ be homogeneous of degree $\abs{A} = k$, with $\tau_{+}(A) = 0$. 
Since $\tau$ factorizes, we have 
$\tau\lrp{\Theta(A)\circ A}= \abs{\tau_{+}(A)}^{2} = 0$.
Insert the expansion \eqref{eq:Hexpansion} into the expression \eqref{eq:necessary} obtained in 
Lemma~\ref{lemma:nobasis},
and use Lemma~\ref{Prop:TProduct-MClosed} to find
\beq \label{eq:insertexpansion}
0 &\leq &\sum_{I,J \in \mathcal{I}}
J_{IJ}\,
 \tau((\Theta(A)\circ A)(\Theta(C_{I})\circ C_{J}))\nn
& = & 
\sum_{I,J \in \mathcal{I}}
J_{IJ}\ \tau\lrp{\Theta(AC_{I})\circ AC_{J}}
= \sum_{I,J \in \mathcal{I}}
J_{IJ}\,\overline{\alpha}_{I} \,\alpha_{J}\;,
\eeq
with $\alpha_{I}=\tau_{+}(AC_{I})$.
In the last expression, we use the fact that $\tau$ factorizes and is reflection invariant. 
Note that $\alpha_{I_{0}} = \tau_{+}(A)$ is zero by assumption, and that 
$\alpha_{I} = 0$ if $\abs{C_{I}} \neq -\abs{A}$ since $\tau$ is neutral.

Since $J_{IJ} = 0$ unless $\abs{C_I} = \abs{C_J}$, 
it suffices to check that 
$0\leq \sum_{I,J} J_{IJ} \overline{\chi}_{I}\chi_{J}$
for every homogeneous vector $(\chi_{I})_{\mathcal{I}}$
which has
finitely many nonzero entries.  Since we are interested in the positivity of the submatrix $(J^{0}_{IJ})_{\mathcal{I}\backslash \{I_{0}\}}$ of couplings across the reflection plane, we can restrict attention to vectors for which ${\chi_{I_{0}}=0}$.
A vector $(\chi_{I})_{\mathcal{I}}$ is called homogeneous of degree $k\in \Z_{p}$ if
every nonzero component $\chi_{I}$ has $\abs{C_{I}} = k$. 

 Let $(\chi_{I})_{\mathcal{I}}$ be a vector as described above, and 
 set $A = \sum_{I\in \cI} \chi_{I}C^{\sharp}_{I}$.  
 For this choice of $A$, we use $\tau_{+}(C^{\sharp}_{I}C_{J}) = \delta_{{IJ}}$
 (assumption B2 in \S\ref{sec:HomOrthBas}) to see that 
 $\alpha_{I} = \tau_{+}(AC_{I})= \chi_{I}$.
 Combining this with \eqref{eq:insertexpansion}, we find that 
 $0\leq \sum_{I,J} J^{0}_{IJ} \overline{\chi}_{I}\chi_{J}$, as required.
 \end{proof}

\subsection{Characterization of RP}\label{sec:Characterization}
Combining the sufficient conditions for reflection positivity in 
Theorem \ref{Thm:sufficient} with the necessary conditions in 
Theorem \ref{thm:necessary}, we obtain the following characterization of reflection positivity.
It holds for any {$q$-double} $\A$ 
satisfying the properties Q1--6, and the further requirements that
the exponential map $\exp \colon \A \rightarrow \A$ is continuous, and 
$\beta \mapsto \exp(-\beta H)$ is differentiable at zero.

\begin{theorem}\label{Thm:neccandsuff}
Let $\tau$ be a continuous, neutral, factorizing  functional on $\A$,
and suppose that $\tau_{+}$ is strictly positive with respect to the map
$\sharp \colon \A_{+} \rightarrow \A_{+}$.  Let 
$H \in \A$ be a reflection invariant operator of degree zero.  
Then the following are equivalent: 
\begin{itemize}
\item[a.] The Boltzmann functional $\tau_{\beta H}$ is reflection positive for all ${0\leq\beta}$.
\item[b.] There exists an $\varepsilon >0$ such that $\tau_{\beta H}$ is reflection 
positive for $0 \leq \beta < \varepsilon$.
\item[c.] The matrix $(J^{0}_{IJ})_{\mathcal{I}}$ of coupling constants across the reflection plane is 
positive semidefinite.
\item[d.] There is a decomposition $H = H_- + H_0 + H_+$, with 
$H_{+} \in \A_{+}$, with $-H_0 \in \cok$, and with $H_- = \Theta(H_+)$.
\end{itemize}
\end{theorem}
\begin{proof}
The implication $a \Rightarrow b$ is clear, and $b \Rightarrow c$ is 
Theorem~\ref{thm:necessary}.
For $c \Rightarrow d$, note that since $H \in \A^0$ is reflection invariant 
and $(J^{0}_{IJ})_{\mathcal{I} \backslash \{I_0\}}$ is positive semidefinite, we
can decompose $H$
as $H = H_+ + H_0 + H_-$  with $-H_0 \in \cok$, $H_+ \in \A^{0}_{+}$, 
and $\Theta(H_+) = H_-$.
The operators $H_0$ and $H_+$ are given in terms of the matrix of coupling constants by 
\beq
	- H_0 &=& \sum_{I,J \in \mathcal{I}\backslash I_{0}} 
	J^{0}_{IJ}\zeta^{\abs{C_I}^2}\Theta(C_I)C_J,\\
	-H_{+} &=& {\textstyle \frac{1}{2}} J_{I_0I_0} \one + \sum_{J \in \mathcal{I}\backslash \{I_0\}} J_{I_0 J}C_{J}\,.
\eeq
Finally, $d\Rightarrow a$ is Theorem~\ref{Thm:sufficient}.
\end{proof}

\begin{remark}
Note the similarity between Theorem~\ref{Thm:neccandsuff} and
Schoenberg's theorem \cite{S38a,S38b}, which states that
$e^{-H}$ is a positive definite kernel on a (discrete) set $\Gamma$ if and only if
$H$ is conditionally negative definite.
Using a limiting argument, we recover Schoenberg's theorem by 
applying Theorem~\ref{Thm:neccandsuff}
to the algebra $\A = C_{c}(\Gamma \times \Gamma)$, with the reflection 
$\Theta(F)(\gamma_{-},\gamma_{+}) = \overline{F}(\gamma_{+},\gamma_{-})$, and 
the algebra $\A_+$ consisting of functions $F(\gamma_-,\gamma_+)$ that depend only 
on $\gamma_{+}$.
\end{remark}

\begin{remark}
In the context of modular theory, it was shown by Connes \cite[Th\'eor\`eme~3.4]{Connes1974}
that $\tau_{\beta H}$ is reflection positive for all negative as well as 
positive $\beta$, if and only if 
$H$ is of the form \eqref{eq:decomposition} with $H_0 = 0$.
By Remark~\ref{rk:beta}, the `if' part of this theorem extends to the $\Z_p$-graded setting.
We recover the `only if' part 
if $\tau$ is factorizing,
which is generally not the case in the context of modular theory.
We consider this an indication that there should exist interesting extensions of 
Theorem~\ref{Thm:neccandsuff} to the case where $\tau$ is not factorizing.
\end{remark}

\setcounter{equation}{0}
\section{Lattice Statistical Physics}\label{sec:latphys}

We illustrate our general framework with an extensive list of examples in the 
context of statistical physics on a lattice.
In this section, we establish fundamental notation that we use in
sections {\S\ref{sec:bosonicsystems}--\ref{sec:parafermionsNo2}}.
%

\subsection{Lattices}\label{sec:lattices}
A lattice is a countable set $\Lambda$, equipped with a reflection 
$\vartheta \colon \Lambda \rightarrow \Lambda$ satisfying $\vartheta^2 = \mathrm{Id}$.
We choose a decomposition $\Lambda = \Lambda_- \cup \Lambda_+$
such that the intersection $\Lambda_0 = \Lambda_+ \cap \Lambda_-$
is the fixed point set of $\vartheta$, and
$\vartheta$ interchanges $\Lambda_+$ with $\Lambda_-$.

To each subset $U \subseteq \Lambda$, we associate 
an algebra $\A_{U}$ of observables.
The algebra corresponding to a single lattice point 
$\lambda \in \Lambda$ is denoted by $\A_{\lambda}$.
The nature of the algebras $\A_{\lambda}$,
as well as their mutual exchange relations inside the algebra $\A_{\Lambda}$ of 
observables associated to the lattice $\Lambda$, 
depends somewhat on the particular situation.  
In the examples below, $\A = \A_{\Lambda}$ will be the $q$-double of 
$\A_{\pm} = \A_{\Lambda_{\pm}}$. 

In most of these examples, we will work with \emph{finite} lattices. 
This captures the essence of the problem; 
if one takes the $C^*$-completion for an infinite lattice, 
reflection positivity carries over to the infinite case. 
We illustrate this in the case of parafermion algebras 
and CPR algebras in \S\ref{sec:parafermionsNo2}, where we treat countable lattices.
We will allow for a nontrivial fixed point set $\Lambda_{0} = \Lambda_+ \cap \Lambda_-$
unless specified otherwise.

\begin{remark}[\bf Reflections in Metric Spaces]
In practice, $\Lambda$ is usually a discrete subset of a metric space 
$\mathcal{M}$, and the reflection comes from an isometry 
$\vartheta_{\mathcal{M}} \colon \mathcal{M} \rightarrow \mathcal{M}$ which `flips'
the ambient space, meaning that $\vartheta_{\mathcal{M}}^2 = \mathrm{Id}$.
In that case, $\Lambda_0 = \Lambda \cap P$ 
is the intersection of $\Lambda$
with the fixed point set 
\[P = \{m\in \mathcal{M}\,;\, \vartheta_{\mathcal{M}}(m) = m\}\,.\]

A typical example is $\mathcal{M} = \R^d$, with
$\vartheta_{\R^d} \colon \R^{d} \rightarrow \R^{d}$ the orthogonal reflection in 
a hyperplane $P \subset \R^{d}$ with unit normal $\hat{n}$, 
and $\Lambda \subseteq \R^{d}$ is a finite 
subset with $\vartheta_{\R^d}(\Lambda) = \Lambda$. 
Then $\Lambda_0 = \Lambda \cap P$
is the intersection of $\Lambda$ with the reflection plane $P$, and 
$\Lambda_{\pm} = \{\lambda \in \Lambda\,;\, \pm\lra{\lambda,\hat{n}}\geq 0\}$
is the part of $\Lambda$ on either side of the reflection plane $P$, 
with points on $P$ included.
\begin{figure}[h!]
  \centering
  \begin{picture}(160,170)(40,15)
   \resizebox{0.75\textwidth}{!}{%
  		\includegraphics[height = 5 cm]{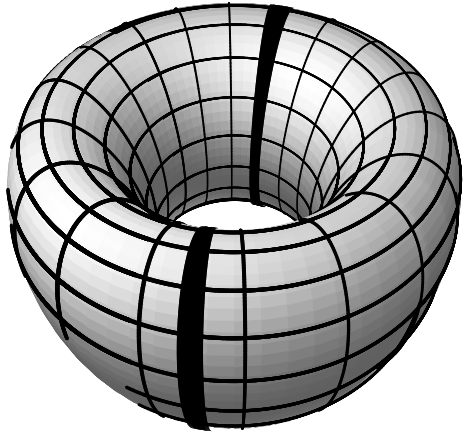}
		}
 \end{picture}
   \caption{{\small Lattice on the torus $T^2$. The fixed point set $P$
   under a reflection is the union 
   of two copies of $T^1$.}}
   \label{fig:QCDlatticeTorus}
\end{figure}

Another common situation is where $\mathcal{M}$ is the $d$-torus 
$T^d = \R^d/L\Z^d$, and $\vartheta_{T^d}(t_1, \ldots, t_n) = (t_1, \ldots, -t_i,\ldots t_n)$
is the reflection in 
one of the coordinates. In that case, the fixed point set
\[P = \{t \in T^d\,;\, t_i \in {\textstyle \frac{1}{2}}L \Z \}\]
is the disjoint union of \emph{two} tori of dimension $d-1$, separated by a distance 
$L/2$ (see Figure~\ref{fig:QCDlatticeTorus}).
\end{remark}

\setcounter{equation}{0}
\section{Bosonic Systems}\label{sec:bosonicsystems}

We specialize our characterization of reflection positivity to bosonic
classical and quantum systems on a finite lattice.

\subsection{Bosonic Classical Systems}\label{sec:BoseClass}
We describe 
an isolated system at a single lattice point $\lambda$
by a probability space $(\Omega, \Sigma, \mu)$.
In the absence of interactions, a bosonic classical system on a lattice $\Lambda$ 
is described by the product
$\Omega^{\Lambda} = \prod_{\lambda \in \Lambda} \Omega_{\lambda}$, with the sigma algebra 
$\Sigma_{\Lambda} = \bigotimes_{\lambda \in \Lambda}\Sigma_{\lambda}$
 and the product measure 
$\mu_{\Lambda}$. 
Denote the sigma algebras on $\Omega^{\Lambda_{\pm,0}}$  by
$\Sigma_{\pm,0} = \bigotimes_{\lambda \in \Lambda_{\pm,0}} \Sigma_{\lambda}$, 
and the corresponding product measures by $\mu_{\pm,0}$.

We now define a reflection $\Theta$ of the algebra 
$L^{\infty}(\Omega^{\Lambda},\mu_{\Lambda})$.
Assume that $\rho \colon \Omega \rightarrow \Omega$ is a reflection for the 
state space of a single system, with $\rho^2 = \mathrm{Id}$ and $\rho_{*}\mu = \mu$.
Then the reflection $\theta \colon \Omega^{\Lambda} \rightarrow \Omega^{\Lambda}$ of the full system $\Omega^{\Lambda}$ is defined 
by $\theta(\omega)_{\lambda} = \rho(\omega_{\vartheta(\lambda)})$.
The antilinear reflection $\Theta \colon L^{\infty}(\Omega^{\Lambda},\mu_{\Lambda}) \rightarrow L^{\infty}(\Omega^{\Lambda},\mu_{\Lambda})$
is
\[\Theta(f)(\omega) = \overline{f(\theta(\omega))}\,.\]

\subsubsection{Reflection Positivity}

We now review the translation of the notion of reflection positivity from algebras 
to measure spaces.
\begin{definition}
A complex valued measure $\nu$ on $\Omega^{\Lambda}$
is called reflection 
positive if 
\be
0 \leq \bE_{\nu}(\Theta(f_{+}) f_+)
\ee
for all 
$f_{+} \in L^{\infty}(\Omega^{\Lambda},\mu_{\Lambda})$ that are measurable w.r.t. $\Sigma_+$. 
\end{definition}
\begin{proposition}\label{prop:primRBclassicalboson}
The measure $\mu_{\Lambda}$ on $(\Omega^{\Lambda}, \Sigma_{\Lambda})$ is reflection positive 
if either  $\rho = \mathrm{Id}$ or $\Lambda_0 = \emptyset$.
\end{proposition}
\begin{proof}
If $\Lambda_{0} = \emptyset$, then since $\mu = \mu_- \otimes \mu_+$
and $\bE_{\mu_-}(\Theta(f_+)) = \overline{\bE_+(f_+)}$,
\[
	\bE_{\mu}(\Theta(f_+)f_+) = \bE_{\mu_-}(\Theta(f_+))\bE_{\mu_+}(f_+)
	= |\bE_{\mu_+}(f_+)|^2
\]
is nonnegative regardless of $\rho$.
If $\Lambda_0 \neq \emptyset$,
let $\lra{f_{+}} \in L^{\infty}(\Omega_{0},\Sigma_{0})$ be the conditional 
expectation of $f_{+}$ with respect to the sigma algebra $\Sigma_{0}$.
Then 
\[
	\bE_{\mu}(\Theta(f_+)f_+) = \int_{\Omega_0} 
	\overline{\lra{f_+}}(\theta(\omega_{0})) \lra{f_+}(\omega_{0})
	\prod_{\lambda \in \Lambda_0}\mu_{\lambda}(d\omega_{\lambda})\,.
\] 
If $\rho = \mathrm{Id}$, then $\theta(\omega_0) = \omega_0$ for all $\omega_0 \in \Omega_0$,
so the above expression is manifestly
nonnegative.
\end{proof}

\subsubsection{Reflection Positivity of Boltzmann Measures}\label{RPclassicalbosonic}
If the interaction is given by a Hamiltonian 
$H \in L^{\infty}(\Omega^{\Lambda},\mu_{\Lambda})$, then
the system at inverse temperature $\beta$ is described by the Boltzmann measure
\be
	\mu_{\beta H} = e^{-\beta H}\mu\,.
\ee
If $H$ is real-valued, then $\mu_{\beta H}$ is a positive measure, which 
can be normalized to the probability measure
$Z(\beta)^{-1}\mu_{\beta H}$, where $Z$ is the partition sum  
$Z(\beta) = \int_{\Omega^{\Lambda}} e^{-\beta H} \mu(d\omega)$.
In this paper, we allow $H$ and $\mu_{\beta H}$ to be complex valued.

We now formulate the necessary and sufficient conditions for 
reflection positivity of $\mu_{\beta H}$.
Fix bounded, square integrable 
elements
 $c_i \in L^{\infty}(\Omega,\mu)$, labelled by $i\in S$, 
with the following properties:
\begin{itemize} 
\item[-] Identity, $c_{i_0} = 1$ for some label $i_0 \in S$.
\item[-] Orthogonality, $\int_{\Omega} \overline{c}_i(\omega) c_j(\omega) \mu(d\omega) = \delta_{ij}$.
\item[-] The linear span of the $c_i$ is dense in $L^{\infty}(\Omega,\mu)$
with respect to the topology of convergence in measure.
\end{itemize}
From this, we obtain a basis of $L^{\infty}(\Omega^{\Lambda_+},\mu_+)$ by
\[
	C_{I}(\omega_+) = \prod_{\lambda \in \Lambda_+}c_{i_{\lambda}}(\omega_{\lambda})\,.
\]
It is labelled by indices $I \in S^{\Lambda_{+}}$. Denote by $I_0$ the index 
that assigns label $i_0$ to every $\lambda \in \Lambda^{+}$. 
Then $C_{I_{0}} = 1\otimes \ldots  \otimes1$ is the identity function, and 
$\bE_{\mu_+}(C_I) = 0$ for $I \neq I_0$.
Further, all $C_I$ are bounded, and their span is dense in 
$L^{\infty}(\Omega^{\Lambda_+}, \mu_{+})$ for the topology of 
convergence in measure.

If $\Lambda_0 = \emptyset$, then we obtain an orthonormal basis of 
$L^{2}(\Omega^{\Lambda}, \mu_{\Lambda})$ labelled by 
$(I,J) \in S^{\Lambda_+} \times S^{\Lambda_+}$,
\[
	B_{IJ}(\omega) =  \Theta(C_{I})C_J \,(\omega) = \prod_{\lambda \in \lambda_+} c_{j_{\lambda}} 
		(\rho(\omega_{\vartheta(\lambda)})) \times  \prod_{\lambda \in \Lambda_{+}}c_{j_{\lambda}}(\omega_{\lambda})\,.
\]
Again, $\Theta(C_{I_0})C_{I_0} = 1$ and 
$\bE_{\mu_{\Lambda}}(\Theta(C_I)C_J) = 0$ for $(I,J) \neq (I_{0},I_0)$. 
As the closure of the $\Theta(C_I)C_J$ is dense in $L^{\infty}(\Omega^{\Lambda},\mu_{\Lambda})$,
every hamiltonian $H\in L^{\infty}(\Omega^{\Lambda},\mu_{\Lambda})$ can be written as
\be\label{eq:doublebasisclassical}
	-H = \sum_{I,J \in S^{\Lambda_+}} J_{IJ} \Theta(C_I)C_J\,,
\ee
where the sum converges in measure, and $(J_{IJ})$ is the \emph{matrix of coupling constants}.
Its submatrix $(J^{0})_{IJ}$ of coefficients with $I, J \neq I_0$ is called the \emph{matrix of coupling constants
across the reflection plane}.
Since $\Theta(\Theta(C_I)C_J) = \Theta(C_J)C_I$, one sees that $H$ is reflection invariant, 
$\Theta(H) = H$, if and only if 
$(J_{IJ})$ is hermitian, $J_{JI} = \overline{J}_{IJ}$.

In the case that the lattice $\Lambda$ does not intersect the reflection plane,
we obtain the following necessary and sufficient conditions for reflection positivity.
\begin{theorem}\label{Thm:classicalwithoutintersection}
Let $H \in L^{\infty}(\Omega^{\Lambda}, \mu_{\Lambda})$ be reflection invariant, 
$\Theta(H) = H$, and suppose that  
$\Lambda_0 = \emptyset$. 
Then the Boltzmann measure 
$\mu_{\beta H} = e^{-\beta H} \mu$ is reflection positive
for all $\beta \geq 0$ if and only if the matrix $J^{0}_{IJ}$ of coupling constants 
across the reflection plane is positive semidefinite.
\end{theorem}

\begin{proof}
Apply Theorem~\ref{Thm:neccandsuff} to the algebra $\A = L^{\infty}(\Omega^{\Lambda}, \mu_{\Lambda})$,
with $\A_{\pm} = L^{\infty}(\Omega^{\Lambda_{\pm}}, \mu_{\pm})$ and 
$\Theta(f)(\omega) = \overline{f(\theta(\omega))}$.
\end{proof}

If $\Lambda_0 \neq \emptyset$, then the expansion 
\eqref{eq:doublebasisclassical} is no longer unique.
Nonetheless, we have the following sufficient conditions for reflection positivity
in the general case, with either $\Lambda_{0} = \emptyset$ or $\rho = \mathrm{Id}$.

\begin{theorem}\label{Thm:classicalwithintersection}
Suppose that
$H = H_{-} + H_0 + H_{+}$, where the element $H_+ \in L^{\infty}(\Omega_{\Lambda_{+}},\mu_+)$
is measurable w.r.t.\ $\Sigma_+$, $\Theta(H_+) = H_-$, 
and $H_0 \in L^{\infty}(\Omega^{\Lambda},\mu_{\Lambda})$ 
possesses an expansion \eqref{eq:doublebasisclassical}
with a positive semidefinite matrix 
$J^{0}_{IJ}$ of coupling constants.
Then the Boltzmann measure 
$\mu_{\beta H} = e^{-\beta H} \mu$ is reflection positive
for all $\beta \geq 0$.
\end{theorem}

Note that Theorems~\ref{Thm:classicalwithoutintersection} 
and~\ref{Thm:classicalwithintersection} allow for couplings between 
arbitrarily many lattice points at arbitrary distance.
We now specialize these results to 
the case of
pair interactions, which is of particular relevance.

\subsubsection{Pair interactions and nearest neighbor interactions}
In this section and the following one, we make some additional assumptions 
on the form of $H$, and we give necessary and sufficient conditions for reflection 
positivity within this class of Hamiltonians.

A \emph{pair interaction Hamiltonian} has the form
\be\label{eq:2-point}
	-H(\omega) = \sum_{\lambda, \lambda' \in \Lambda} 
	h_{\lambda\lambda'}(\omega_{\lambda},\omega_{\lambda'})	+
	 \sum_{\lambda \in \Lambda} V_{\lambda}(\omega_{\lambda})\,.
\ee
For general pair interactions, we do not impose any restrictions on the finite lattice $\Lambda$
other than the ones in \S\ref{sec:lattices}.

A Hamiltonian $H$ is of \emph{nearest neighbor type}
if it is of the form \eqref{eq:2-point} with
$h_{\lambda\lambda'}$ nonzero only for $\abs{\lambda - \lambda'} = 1$.
We have special results for Hamiltonians describing 
nearest neighbor interactions on rectangular lattices
in $\R^d$ or $T^d = \R^d/(L\Z)^d$, of the form
\beq\label{eq:reclattice}
	\Lambda = \{-L, \ldots, L\}^{d} \subseteq \R^d
	\quad \text{or} \quad
	\Lambda = \{0, \ldots, L\}^{d} \subseteq T^d\,.
\eeq
Here, we assume that the fixed point set $P\subseteq \R^d$ is in one of the coordinate planes,
and that it intersects the lattice nontrivially.


\begin{theorem} Suppose that $\Lambda$ is a rectangular lattice 
of the form \eqref{eq:reclattice},
intersecting the coordinate plane $P$ nontrivially.
Let $\theta(\omega_{\lambda}) = \omega_{\vartheta(\lambda)}$.
Then for every 
reflection invariant nearest neighbor hamiltonian 
$H \in L^{\infty}(\Omega^{\Lambda},\mu_{\lambda})$, the Boltzmann measure
$\mu_{\beta H}$ is reflection positive for all $\beta \geq 0$.
\end{theorem}
\begin{proof}
Nearest neighbor Hamiltonians on a lattice 
that intersects the reflection plane are very special, 
since they allow a decomposition $H = H_- + H_{0} + H_+$ 
with $H_{0} = 0$.

To see this, note that each bond $\lra{\lambda,\lambda'}$ is contained in 
either~$\Lambda_+$ or~$\Lambda_-$. We can thus write the hamiltonian as 
$H = H_- + H_+$ where $H_+$ is measurable w.r.t.\ $\Sigma_+$, and 
$H_- = \Theta(H_+)$. 
The corollary then follows from Theorem \ref{Thm:classicalwithintersection}.
To exhibit the splitting, 
define $H_+$ by
\be 
	-H_{+} = \sum_{\lambda,\lambda' \in \Lambda_+} 
		\epsilon_{\lambda\lambda'}h_{\lambda\lambda'}(\omega_{\lambda},\omega_{\lambda'})
		+ 
		\sum_{\lambda \in \Lambda_+} 
		\epsilon_{\lambda}V_{\lambda}(\omega_{\lambda})\,,
\ee
with $\varepsilon_{\lambda\lambda'} = \frac{1}{2}$ if both
$\lambda$ and $\lambda'$ are in $\Lambda_0$, and 
$\varepsilon_{\lambda\lambda'} = 1$ otherwise.
Similarly, $\epsilon_{\lambda} = \frac{1}{2}$ if $\lambda \in \Lambda_0$
and $\epsilon_{\lambda} = 1$ if $\lambda \in \Lambda_+ \backslash\Lambda_0$.
As $H$ is reflection invariant, it can be written in the form \eqref{eq:2-point}
with $h_{\vartheta(\lambda),\vartheta(\lambda')} = \overline{h}_{\lambda,\lambda'}$
and $V_{\vartheta(\lambda)} = \overline{V}_{\lambda}$.
Using this, one verifies that $H = \Theta(H_+) + H_{+}$, as required.
\end{proof}

Nearest neighbor interactions on a lattice that does \emph{not} intersect the reflection plane are \emph{not} automatically reflection positive. They are characterized in
Remark~\ref{rk:nnnointer}


\subsubsection{Long Range Pair Interactions}
Suppose that for each lattice site $\lambda$, we have $k$ 
random variables $\phi^{a} \in L^{\infty}(\Omega,\mu)$, 
with $a = 1, \ldots, k$. We require that
$\bE_{\mu}(\phi^a) = 0$ and $\bE_{\mu}(\phi^a\phi^b) = \delta_{ab}$.
For example, if $\Omega$ is the 2-point space $\{+1,-1\}$ 
with the counting measure, one can take the single  
variable $\phi(\omega) = \omega$.
If $\Omega$ is the block 
$\Omega = [-\phi_{\mathrm{max}},\phi_{\mathrm{max}}]^{k}$ with the normalized 
Lebesque measure, or the $k-1$-sphere $\Omega = S^{k-1} \subseteq \R^{k}$
with the round measure,
then one can take 
$\phi^{a}$ to be the coordinate variables. 
Consider Hamiltonians of the form~\ref{eq:2-point}
with 
\be\label{eq:fieldscor}
	h_{\lambda,\lambda'} = J^{ab}_{\lambda,\lambda'} \phi_{\lambda}^{a}\phi_{\lambda'}^{b}\,,
\ee
where the reflection $\rho$ sends $\phi^{a}$ to $s^a \phi^a$, with $s^a = \pm 1$.
If the reflection plane $P \subseteq \R^d$ does not intersect the lattice,
then necessary and sufficient conditions for reflection positivity can be given as follows.

The matrix of couplings 
across the reflection plane is
$(s^{a}J^{ab \,\,\,0}_{\vartheta(\lambda),\lambda'})$, 
with entries labelled by 
$(\lambda,a)$ and $(\lambda',b)$ in $\Lambda_{+} \times \{1, \ldots, k\}$.
The following corollary then follows immediately from
Theorem~\ref{Thm:classicalwithoutintersection}. 

\begin{corollary}
Suppose that $H$ is a reflection invariant Hamiltonian 
of the form \eqref{eq:2-point}, with 
$h_{\lambda\lambda'}$ given by~\eqref{eq:fieldscor}.
Then $\mu_{\beta H}$ is 
reflection positive for all $\beta \geq 0$ if and only if
$(s^aJ^{ab \,\,\,0}_{\vartheta(\lambda)\lambda'})$ 
is positive semidefinite. 
\end{corollary}

\begin{remark}[Nearest Neighbor]\label{rk:nnnointer}
Consider a rectangular 
lattice 
\[\textstyle
	\Lambda = \{-\frac{2L+1}{2}, -\frac{2L-1}{2}, \ldots, \frac{2L+1}{2} \}^{d}
\]
that does not intersect the reflection plane $P$, and a nearest neighbor Hamiltonian given by
$J^{ab}_{\lambda\lambda'}$. Then $H$
is reflection invariant, if for every bond $\lra{\theta(\lambda),\lambda}$ that crosses the reflection plane,
the $k\times k$-matrix
in $a$ and $b$ given by $(s^aJ^{ab \,\,\,0}_{\theta(\lambda),\lambda})$ is positive semidefinite.
\end{remark}

More generally, if $\Lambda \subseteq \R^d$ is any lattice that does not intersect $P$,
and $H$ is a reflection invariant Hamiltonian with 
$J_{\lambda,\lambda'}^{ab} = f(\lambda - \lambda')J^{ab}$, then 
$\mu_{\beta H}$ is reflection invariant for all $\beta \geq 0$ if and only if 
$s^a s^b J^{ab}$ is positive semidefinite, and $f \colon \R^{d}\rightarrow \R$
is \emph{OS-positive}, 
\be 
\sum_{i,j=1}^{n}\overline{z}_iz_j f(\vartheta(\lambda_i) - \lambda_j) \geq 0
\ee
for all $(z_i, \lambda_i) \in \C \times \R^{d, +}$. (OS stands for Osterwalder-Schrader.) 
For example, the function 
$f(\lambda) = \abs{\lambda}^{-s}$ is OS-positive if $s \geq d-2$ and $s\geq 0$.
Naturally, we have a similar sufficient condition for reflection positivity 
in case that $\Lambda_0 \neq \emptyset$. The only difference is that 
all signs $s^a$ equal $+1$, as $\rho$ must be the identity.

\subsection{Bosonic Quantum Systems}\label{sec:BoseQuant}
Suppose that the isolated system at each lattice point $\lambda$
is a bosonic, quantum mechanical system with $n$ degrees of freedom. This is described by the matrix algebra $\A_{\lambda} = M_{n}(\C)$. The total system is given by the algebra 
\[
	\A = \bigotimes_{\lambda \in \Lambda} M_{n}(\C)\,.
\]
In the absence of interactions, the background state $\tau$ is 
the normalized tracial state, given by
\be\label{eq:deftrace}
	\tau(A_{\lambda_1} \otimes \ldots \otimes A_{\lambda_k}) = \frac{1}{n^k}\tr(A_{\lambda_1})
	\cdots \tr(A_{\lambda_k})
\ee
on the pure tensors. (Here $\mathrm{Tr}$ denotes the unnormalized trace.)
The reflection $\Theta \colon \A \rightarrow \A$ is the antilinear 
homomorphism given by
\be\label{eq:matrixflip}
	\Theta(A_{\lambda}) = \overline{\rho(A)}_{\vartheta(\lambda)}\,,
\ee
where $\overline{\phantom{a}}$ denotes complex conjugation and $\rho$  denotes conjugation by an arbitrary invertible operator $R\in \mathrm{GL}_{n}(\C)$, namely $\rho(A) = RAR^{-1}$.  
If $R$ is unitary, then $\rho(A^*) = \rho(A)^*$, but we will 
\emph{not} require that this is the case.

For bosonic quantum systems, we only consider the case
${\Lambda_{0} = \emptyset}$,
meaning that the reflection $\vartheta$ has no fixed points on $\Lambda$.
If we define 
\[\
	\A_{\pm} = \bigotimes_{\lambda \in \Lambda_{\pm}} M_{n}(\C)\,,
\]
then $\Theta(\A_+) = \A_{-}$, and 
$\A$ is the linear span of 
$\A_- \A_+$. Since $\A = \A_- \otimes \A_+$, the algebra 
$\A$ is the bosonic $q$-double of $\A_+$ (cf. \S\ref{sec:tensprod}).


\begin{proposition}[\bf Primitive Reflection Positivity]
The tracial state $\tau$ is faithful, factorizing, reflection invariant, and reflection positive; 
\be
0 \leq \tau(\Theta(A)A)\;, 
\quad\text{for all}\quad A\in\A_{+}.
\ee

\end{proposition}
\begin{proof}
By linearity, it suffices to show reflection invariance
on the pure tensors 
$
A = A_{\lambda_1} \otimes \ldots \otimes A_{\lambda_{k}}
$.
This follows from the identity
\be
\tau(\Theta(A)) = {\textstyle \frac{1}{n^k}}
\tr(\overline{RA_{\lambda_1}R^{-1}})\cdots\tr(\overline{RA_{\lambda_k}R^{-1}})
=\overline{\tau(A)}\,.
\ee
By \eqref{eq:altfactor}, the factorization property can be expressed as 
$
\tau(A B) = \tau_{-}(A)\tau_{+}(B)
$
for $A\in \A_{-}$ and $B \in \A_{+}$. This
is immediate from 
\eqref{eq:deftrace}.
The state $\tau$ is faithful since it is a finite tensor product of faithful states, and 
reflection positive by Proposition~\ref{invimppos}.
%
\end{proof}
Fix an orthonormal basis $\{c_i\}_{{i\in S}}$ of $M_{n}(\C)$ with respect to the inner product 
$(X,Y) = \Tr(X^*Y)$ such that $c_0 = \one$.  
The basis is labelled by $i \in S$.
A usual choice is the basis consisting of
$\one$, the matrices $E_{kk} - E_{k+1,k+1}$ for 
$k = 1, \ldots, n-1$, and the matrices 
$E_{kl} + E_{lk}$ and 
$i(E_{kl} - E_{lk})$ for $1\leq k< l \leq n$. 
Here $E_{kl}$ denotes the matrix with entry~$1$ in the $kl$ place and $0$ elsewhere.  
In the case of $M_2(\C)$, these are the 
Pauli matrices.  

From the basis $\{c_i\}_{i\in S}$ for $M_n(\C)$, we obtain the tensor product basis
\[
C_{I} = \bigotimes_{\lambda \in \Lambda_+} c_{i_{\lambda}}
\]
for $\A_{+}$, labelled by $I \in S^{\Lambda_{+}}$. In turn, this yields the basis
\[
	B_{IJ} = \Theta(C_I) \circ C_J 
	= \Theta(C_I) C_J 
	= \bigotimes_{\kappa \in \Lambda_-} 
	\overline{Rc_{i_{\theta(\kappa)}}R^{-1}} \bigotimes_{\lambda \in \Lambda_+} c_{j_{\lambda}}
\]
of $\A$, labelled by $(I,J) \in S^{\Lambda_+} \times S^{\Lambda_{+}}$.
Every matrix $H\in \A$ has a unique expansion
\be
-H = \sum_{I,J \in S^{\Lambda_+}} J_{IJ} B_{IJ}
\ee
in the basis $B_{IJ}$. The basis coefficients form an 
$S^{\Lambda_+} \times S^{\Lambda_+}$-matrix 
$(J_{IJ})_{S^{\Lambda_+}}$, called the 
\emph{matrix of coupling constants}.
The \emph{matrix of coupling constants across the reflection plane}
is the submatrix $(J^{0}_{IJ})_{S^{\Lambda_+}\backslash \{0\}}$ where neither 
$C_I$ nor $C_J$ is the identity. 

The following theorem gives necessary and sufficient conditions 
for reflection positivity of the Boltzmann functional 
$\tau_{\beta H} \colon \A \rightarrow \C$ at inverse temperature $\beta \geq 0$, 
defined by $\tau_{\beta H}(A) = \tau(A\,e^{-\beta H})$.
 \begin{theorem}\label{thm:spinsthm}
 Let $H\in \A$ be reflection invariant, $\Theta(H) = H$. Then
 the Boltzmann functional $\tau_{\beta H}$ is reflection positive on $\A_+$ for all 
 $\beta \geq 0$ if and only if the matrix $(J^0_{IJ})_{S^{\Lambda_+}\backslash \{0\}}$ 
 of coupling constants across the reflection plane
 is positive semidefinite.
 \end{theorem}
 \begin{proof}
This follows from Theorem~\ref{Thm:neccandsuff}.
\end{proof}

%
This result  extends \cite[Theorem 5.2]{JaffeJanssens2016} from $M_2(\C)$ to $M_n(\C)$.
As a simple example of how Theorem~\ref{thm:spinsthm} may be used in a concrete situation, 
we show that the long range antiferromagnetic Heisenberg model 
is reflection positive at arbitrary spin $s$, see 
\cite{DysonLiebSimon1976,FILS78,DysonLiebSimon1978}.
The Hamiltonian is
\[
	-H = J \sum_{\lambda \neq \lambda' } \abs{\lambda - \lambda'}^{-v} 
	\sum_{a = x,y,z}
	S^{a}_{\lambda} \,S^{a}_{\lambda'}\,,
\]
where $S^{x}, S^{y}, S^{z} \in M_{2s+1}(\C)$ are hermitian spin matrices for spin $s$. 
In the highest weight representation 
$\pi \colon \mathfrak{sl}(2) \rightarrow M_{2s+1}(\C)$, these are given by
\[{\textstyle
S^{x} = \frac{1}{2}(\pi(e)+\pi(f)), \quad S^{y} = -\frac{i}{2}(\pi(e)-\pi(f)), \quad\text{and} \quad
S^z = \frac{1}{2}\pi(h)\,,
}
\]
where $e$, $h$ and $f$ are the usual $\mathfrak{sl}(2)$-generators 
with 
\[ [h,e] = 2e, \quad [h,f] = -2f,\quad \text{and}\quad [e,f] = h\,.\]
Since $\pi(e)$, $\pi(h)$ and $\pi(f)$ can be realized as real matrices,
the map $X \mapsto \overline{X}$ flips the sign of $S^y$, while leaving 
$S^x$ and $S^z$ invariant. Since $R = \exp( i\pi S^y)$
represents a $180^{\circ}$-rotation around the $y$-axis, the map
$X \mapsto RXR^{-1}$ flips the sign of $S^x$ and $S^z$, while leaving 
$S^y$ invariant.  

With the reflection $\Theta$ of \eqref{eq:matrixflip}, we therefore find  
$\Theta(S^{a}_{\lambda}) = -S^{a}_{\vartheta(\lambda)}$.
The matrix of coupling constants across the reflection plane 
is thus given by 
\[
	J^{0}{}^{ab}_{\lambda\lambda'} = -J\abs{\vartheta(\lambda) - \lambda'}^{-v}\delta^{ab}
\]
for $\lambda,\lambda' \in \Lambda_{+}$, and $a,b \in \{x,y,z\}$.
If $\Lambda$ is a $\vartheta$-invariant subset of $\R^d$, then 
this is a positive semidefinite matrix if $J \leq 0$ and if $v$ is a nonnegative
number with $v \geq d-2$.

\setcounter{equation}{0}
\section{Fermionic Systems}\label{sec:fermionicsystems}

We specialize our characterization of reflection positivity to fermionic classical and 
quantum systems on a lattice.

\subsection{Fermionic Classical Systems}\label{sec:fermpsi}
A  fermi\-onic classical system is described by 
the $\Z_{2}$-graded \emph{Grassmann algebra} $\A = \bigwedge V$,
which we have already considered in \S\ref{ExIINumber1}.
Here $V$ is an oriented, \emph{even-dimensional} Hilbert space, which
may arise either from a single site $\lambda$, or from the full lattice $\Lambda$.


For applications in 
physics, 
the vector space $V$ corresponding to a single site 
is either 
$V = W$ (for Weyl spinors) or $V = W \oplus \overline{W}$ (for Dirac spinors).
In the latter case,
$\overline{W}$ is identified with $W$ by means of an 
antilinear isomorphism $\psi \mapsto \overline{\psi}$.
Here
$W = W_{s}\otimes W_{D}$ is the tensor product 
of an $s$-dimensional, unitary representation $W_{s}$ 
for $\mathrm{spin}(d)$ 
and a $D$-dimensional  
unitary representation $W_{D}$ of the relevant gauge group~$G$.
The basis elements are then labelled by $\psi_{\alpha a}$, with 
$\alpha = 1, \ldots, s$ and $a = 1, \ldots, D$.

The vector space corresponding to the full 
lattice $\Lambda$ is $V^{\Lambda}$, and the algebra
is $\A = \bigwedge V^{\Lambda}$.
The definition of the algebras $\A_-$ and $\A_+$
depends on whether the intersection $\Lambda_0$
of $\Lambda_-$ and $\Lambda_+$  is empty or not.
In case $\Lambda_{0} = \emptyset$, we simply define
$V_{\pm} = V^{\Lambda_\pm}$, and set 
$\A_{\pm} = \bigwedge V^{\Lambda_{\pm}}$.
We allow $\Lambda_0 = \Lambda_- \cap \Lambda_+$ to be nonempty
only if $V = W \oplus \overline{W}$.
In that case, we set 
\be
\begin{aligned}
V_{+} &= \bigoplus_{\lambda \in \Lambda_0} W_{\lambda} \oplus 
\bigoplus_{\lambda \in \Lambda_{+}} (W_{\lambda} \oplus \overline{W}_{\lambda})\,,\\
V_{-} &= \bigoplus_{\lambda \in \Lambda_0} \overline{W}_{\lambda} \oplus 
\bigoplus_{\lambda \in \Lambda_{-}} (W_{\lambda} \oplus \overline{W}_{\lambda})\,,
\end{aligned}
\label{eq:defofVpm}
\ee
and we define $\A_{\pm} = \bigwedge V_{\pm}$.

Let $\rho \colon V \rightarrow V$ be an antilinear isomorphism that squares to the 
identity. If $V = W \oplus \overline{W}$, we require that 
$\rho$ interchanges $W$ and $\overline{W}$.
The reflection $\Theta \colon \A \rightarrow \A$
is the unique antilinear homomorphism
such that
\[
	\Theta(\psi_{\lambda}) = \rho(\psi_{\vartheta(\lambda)})
\]
for all $\psi_{\lambda} \in V_{\lambda}$.  
%
Note that the Grassmann algebra $\A$ is the fermionic $q$-double of $\A_+$, 
cf.~ \S\ref{ExIINumber1}.
%

%
%

\begin{proposition}[\bf RP of the Berezin integral]\label{prop:primfunc}
Suppose that $V$ is even dimensional, and that $\rho(\mu) = \mu$.
If $V = W \oplus \overline{W}$, then we require that
$W$ is even dimensional, 
and that
the restriction of $\overline{\rho}$
to $W \rightarrow W$ is of determinant~$1$.
Then the  Berezin integral is a factorizing, 
reflection invariant, reflection positive functional of degree zero.
\end{proposition}
\begin{proof}
Note that an orientation of $W$ defines an orientation on 
$V^{\Lambda}$, $V^{\Lambda_{+}}$ and $V^{\Lambda_{-}}$.
A positively oriented volume $\mu$ or $\mu_{\pm}$ is obtained by taking the 
product of $\overline{\mu}_{W,\lambda}$ and $ \mu_{W,\lambda}$
over all the sites $\lambda$ in the relevant lattice.
If $\lambda \in \Lambda_{0}$, then
$\mu_{+}$ only gets a single factor $ \mu_{W,\lambda}$, and 
$\mu_{-}$ only gets a single factor 
$\overline{\mu}_{W,\lambda}$.

Since the relevant vector spaces are even dimensional, the order of the products 
is immaterial. The assumptions on $\rho$ ensure that $\theta \colon V_{+} \rightarrow V_{-}$ is volume preserving, and that $\mu = \mu_{-} \wedge \mu_{+}$.
The result then follows from Proposition~\ref{prop:BerezinFact}.
\end{proof}

We construct a basis of $\A$ that is adapted to the reflection.
First, choose a basis $\{\psi_i\}_{i\in T}$ of $W$ and $\{\psi_{i}\}_{i\in S}$ 
of $V$.
From this, we obtain a basis $\psi_{(\lambda, i)}$ of $V_+$, labelled by
$(\lambda,i) \in (\Lambda_0 \times T) \sqcup (\Lambda_+ \backslash\ \Lambda_0) \times S$.
By choosing an order on this label set, we obtain an ordered basis of 
$\A_+ = \bigwedge V_+$ by setting
\be\label{eq:fermionbasis}
C_I = \psi_{(i_1,\lambda_1)} \wedge \ldots \wedge \psi_{(i_k,\lambda_k)}\,,
\ee
if $(\lambda_1,i_1) < \ldots < (\lambda_k,i_k)$ is in increasing order.
The basis $C_I$ is labelled by the power set 
\be\label{eq:indexsetclfermi}
\mathcal{I} = \mathcal{P}\Big((\Lambda_0 \times T) \sqcup (\Lambda_+ \backslash\ \Lambda_0) \times S\Big)\,.
\ee
If $I_0 = \emptyset$, we define $C_{I_0} = \one$ to be the identity.
Using the basis $C_I$ of $\A_+$, we obtain a basis $B_{IJ}$ of $\A^{0}$ by 
\[
	B_{IJ} = \sqrt{-1}^{\abs{I}^2}\Theta(C_I)C_J\,,
\]
where $\abs{I} \in \Z_2$ is the cardinality of $I$ modulo 2, and
$J$ is restricted to have $\abs{I} = \abs{J}$ modulo 2. 
This ensures that $B_{IJ}$ is even.
(Note that the factor $\sqrt{-1} = \zeta$ comes from the twisted product~\eqref{eq:TwistedProduct}.)

Every $H\in \A^0$ then has a basis expansion
\[
	-H = \sum_{I,J} J_{IJ} B_{IJ}\,.
\]	
The matrix $(J_{IJ})_{\mathcal{I}}$ is called
the \emph{matrix of coupling constants}, and 
the submatrix $(J^0_{IJ})_{\mathcal{I} \backslash \{I_0\}}$
is called the 
\emph{matrix of coupling constants across the reflection plane}.

\begin{theorem}
Let $H \in \A$ be a reflection invariant element of degree zero. 
Then the Boltzmann functional 
$\tau_{\beta H}(A) = \tau(A\,e^{-\beta H})$ is reflection positive on $\A_+$
for all $\beta \geq 0$, if and only if 
the matrix of coupling constants across the reflection plane
is positive semidefinite.
\end{theorem}
\begin{proof}
This follows from Theorem~\ref{Thm:neccandsuff}.
The Berezin integral is strictly positive by
Proposition~\ref{Prop:berezinstrict}, and it is factorizing and reflection positive 
by Proposition~\ref{prop:primfunc}.
\end{proof}

\subsection{Fermionic Quantum Systems}\label{sec:FermiQuant}
Quantum mechanical fermi\-onic systems are described by 
\emph{Clifford algebras}, which we
considered in \S\ref{sec:Clifford}.
If the vector space associated to a single lattice site is the finite dimensional 
vector space $V$, then 
the space associated to the full lattice is $V^{\Lambda}$.
Correspondingly, the algebra for a single site is 
$\cA_{\lambda} = \mathrm{Cl}(V)$, and the algebra for the full lattice is 
$\cA = \mathrm{Cl}(V^{\Lambda})$.

Let $\rho \colon V\rightarrow V$ be an antilinear map 
that squares to the identity, and satisfies 
$h_{\C}(\rho(v),\rho(v')) = \overline{h_{\C}(v,v')}$ fot all $v,v' \in V$.
This yields an antilinear isomorphism 
$\theta \colon V^{\Lambda} \rightarrow V^{\Lambda}$
by $\theta(v_{\lambda}) = \rho(v_{\theta(\lambda)})$, and hence
an antilinear homomorphism $\Theta \colon \cA \rightarrow \cA$.

If $\Lambda_0 = \Lambda_+ \cap \Lambda_-$ is nonzero, then 
we require that 
$V_{\R} = W_{\R} \oplus W_{\R}$ is an orthogonal direct sum,
and that $\rho$ is the antilinear complexification of a real orthogonal 
transformation $\rho_{\R} \colon V_{\R} \rightarrow V_{\R}$
that interchanges the two copies of $W_{\R}$.

Define the vector spaces $V_{\pm}$ as in \eqref{eq:defofVpm}, and define 
$\cA_{\pm} = \mathrm{Cl}(V^{\Lambda_{\pm}})$.
Since $\theta(V_+) = V_-$ and $h_{\C}(V_+,V_-) = \{0\}$,
the algebra $\cA$ is the fermionic $q$-double of $\cA_+$, cf.~\S\ref{sec:Clifford}.

Choose orthonormal bases $\{c_i\}_{i\in T}$ of $W_{\R}$, and 
$\{c_i\}_{i\in S}$ of $V_{\R}$. 
In the same way as in \S\ref{sec:fermpsi}, we obtain a basis 
$C_I$ of $\A_+$, labelled by the index set 
$\cI$ of equation~\eqref{eq:indexsetclfermi}. 
It is given by $C_{I_0} = \one$ if $I_0 = \emptyset$, and by
\be
	C_I = c_{(\lambda_1,i_1)}\cdots c_{(\lambda_k,i_k)}\,,
\ee
if $(\lambda_1,i_1) < \ldots < (\lambda_k,i_k)$ is increasing 
with respect to a chosen order on
$(\Lambda_0 \times T) \sqcup (\Lambda_+ \backslash\ \Lambda_0) \times S$.

Using the basis $C_I$ of $\cA_+$, we define 
the basis $B_{IJ}$ of $\cA^{0}$ by 
\[
	B_{IJ} 
	= \Theta(C_I) \circ C_J 
	= \sqrt{-1}^{\abs{I}^2}\Theta(C_I)C_J\,.
\]
Here $\abs{I} \in \Z_2$ denotes the cardinality of $I$ modulo 2, 
and $J \in \cI$ is restricted to have 
$\abs{J} = \abs{I}$. Every $H\in \A^0$ then has a basis expansion
$-H = \sum_{I,J} J_{IJ} B_{IJ}$.  Explicitly, 
there exist unique coefficients 
\[
J_{IJ} = 
J^{i_1}_{\lambda_1} \cdots {}^{i_k\,;}_{\lambda_k\,;}\,{}^{i'_1}_{\lambda'_1}
\cdots {}^{i'_{k'}}_{\lambda'_{k'}}
\]
such that
\[
	-H = 
	\sum J_{IJ} \sqrt{-1}^{\,k^2}\rho(c_{\vartheta(\lambda_1)i_1}) \cdots \rho(c_{\vartheta(\lambda_k)i_k})\,
	c_{\lambda'_1 i'_1} \cdots c_{\lambda'_{k'} i'_{k'}}\,.
\]
The matrix $(J_{IJ})_{\cI}$ is
the \emph{matrix of coupling constants}, and 
$(J^0_{IJ})_{\cI \backslash \{I_0\}}$
is the 
\emph{matrix of coupling constants across the reflection plane}.

\begin{theorem}
Let $\tau \colon \cA \rightarrow \C$ be the tracial state of Definition \ref{ref:deftracialcliff},
and let $H \in \cA$ be a reflection invariant element of degree zero. 
Then the Boltzmann functional 
$\tau_{\beta H}(A) = \tau(A\,e^{-\beta H} )$ is reflection positive on $\cA_+$
for all ${\beta \geq 0}$, if and only if 
the matrix of coupling constants across the reflection plane
is positive semidefinite.
\end{theorem}
\begin{proof}
This follows from Proposition~\ref{prop:cliffordtau} and Theorem~\ref{Thm:neccandsuff}.
 \end{proof}

\setcounter{equation}{0}
\section{Lattice Gauge Theories: Equivariant Quantization}\label{sec:latticeQCD}

We give a characterization of reflection positivity in the context of lattice gauge theories.  
In particular, this yields a new, gauge equivariant proof for reflection positivity of 
the functional determined by the Wilson action \eqref{eq:YMplaquette}, 
stated in Corollary~\ref{cer:rpwilson}.
Wilson introduced this action to be gauge invariant and have the correct pointwise 
continuum limit.
By a miracle, this action also gives a reflection-positive expectation.

In contrast to the proofs in the literature, pioneered 
by Osterwalder and Seiler \cite{O76, OS78, S82}, we do \emph{not} fix the gauge 
on bonds that cross the reflection plane.
Rather, we introduce extra degrees of freedom that put the 
interaction across the reflection plane in a form covered by  
Theorem~\ref{Thm:neccandsuff}. We deal with the problem of fermion doubling as in
\cite{MP87}.

Using this method, we are able to prove reflection positivity
on the full algebra of observables, not just on the gauge invariant part.
As a consequence of the quantization procedure, any two elements of $\A_+$ 
that differ by a gauge transformation that does not involve the reflection plane 
yield the same state in $\cH_{\Theta}$.

%

%

 
\subsection{Gauge Bosons} 
Let $G$ be a compact Lie group, and let
$\Lambda' $ be a hypercubic lattice of width $r$ in $\R^d$ or $T^d$.
Let $\Lambda''$ be the set of midpoints $\lambda'' = \frac{1}{2}(\lambda'_1 + \lambda'_2)$ of 
nearest neighbors $\lambda_1,\lambda_2$ in $\Lambda'$,
and define the lattice as $\Lambda = \Lambda' \cup \Lambda''$.
(See Figure~\ref{fig:QCDlattice}.)

Denote the set of directed nearest-neighbor bonds in $\Lambda$ by 
$E = \{\lra{\lambda\lambda'}\,;\, \abs{\lambda - \lambda'} = r/2\}$, 
and denote the set of undirected bonds by 
$\abs{E} = \{\{\lambda\lambda'\}\,;\, \abs{\lambda - \lambda'} = r/2\}$.
To describe the bosonic degrees of freedom, we associate 
the variable $h_{\lra{\lambda\lambda'}} \in G$ to the directed nearest-neighbor bond 
$\lra{\lambda\lambda'} \in E$.
Note that every nearest-neighbor bond contains one site in $\Lambda'$, 
and one in $\Lambda''$.
Since $h_{\lra{\lambda\lambda'}}$ represents the holonomy induced by parallel transport
from $\lambda$ to $\lambda'$, we impose 
$h_{\lra{\lambda'\lambda}} = h_{\lra{\lambda\lambda'}}^{-1}$
for the bond $\lra{\lambda'\lambda}$
in the other direction.
The Haar measure $\mu_H$ on the 1-bond probability space 
\[\Omega_{\{\lambda\lambda'\}} = 
\{(h_{\lra{\lambda\lambda'}}, h_{\lra{\lambda'\lambda}}) \in G \times G\,;\, h_{\lra{\lambda\lambda'}}
= h_{\lra{\lambda'\lambda}}^{-1}\}\]
is obtained from the Haar measure on $G$ by either one of the two projections 
to $G$ (the result is the same).
Similarly, the configuration space of discrete holonomies for the full system is
\[
G^{\abs{E}} := \{h \in G^{E}\,;\, h_{\lra{\lambda\lambda'}} = h^{-1}_{\lra{\lambda'\lambda}}\}\,.
\]
We equip it with the Haar measure obtained from the identification 
$G^{\abs{E}} \simeq \prod_{\{\lambda\lambda'\} \in \abs{E}} \Omega_{\{\lambda\lambda'\}}$. 
The associated algebra of bosonic observables is $\A^{B} = L^{\infty}(G^{\abs{E}})$.
\begin{figure}[h!]
   \centering
   \resizebox{0.6\textwidth}{!}{%
  		\includegraphics[height = 3 cm]{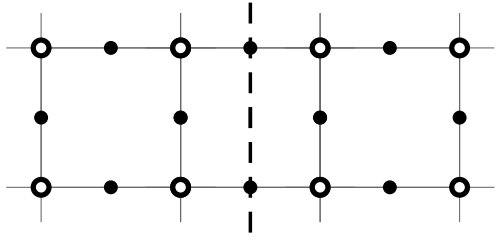}
		}
   \caption{{\small Points in $\Lambda'$ are white, points in $\Lambda''$ are black.}}
   \label{fig:QCDlattice}
\end{figure}

The algebra $\A^{B}$ contains functions that depend on the holonomies between 
every pair $\lra{\lambda\lambda'}$ of nearest neighbors in $\Lambda$.
If we zoom out and consider only the `coarse' lattice $\Lambda' \subseteq \Lambda$
(the white points in Figure~\ref{fig:QCDlattice}),
then the holonomy between the nearest neighbors 
$\lambda, \lambda'$ in $\Lambda'$
is given by $h^{\Lambda'}_{\lambda\lambda'} = 
h_{\lambda\lambda''}h_{\lambda''\lambda'}$, where 
$\lambda'' \in \Lambda''$ is the midpoint between $\lambda$ and $\lambda'$.
Define  
\[
\A^{B}_{\Lambda'} \subseteq \A^{B}
\] 
to be the subalgebra of measurable 
functions that depend only on the variables $h^{\Lambda'}_{\lambda'\kappa'}$.

\subsubsection{Reflection Positivity}\label{sec:RPgluons}
Suppose that the reflection $\vartheta \colon \Lambda \rightarrow \Lambda$ 
flips a single coordinate $x^{\sigma}$.
Then the reflection $\Theta \colon \A^{B} \rightarrow \A^{B}$ 
is the anti-linear homomorphism given by
\beq\label{eq:gaugereflection}
	\Theta(F)(h_{\lra{\lambda,\lambda'}}) &=& \overline{
		     F(h_{\lra{\vartheta(\lambda),\vartheta(\lambda')}})
		}
\eeq
for all $F \in \A^{B}$.

We assume that the fixed point set $P$ is orthogonal to the basis vector 
$\vec{e}_{\sigma}$, and intersects the lattice $\Lambda$ halfway 
between lattice points in $\Lambda'$, so
$P \cap \Lambda'' = \Lambda_{0}$. (See Figure \ref{fig:QCDlattice}.)
It follows that
$\Lambda = \Lambda_{-} \cup \Lambda_+$ with $\Lambda_- \cap \Lambda_+ = \Lambda_0$,
$\Lambda' = \Lambda'_- \cup \Lambda'_{+}$ with 
$\Lambda'_- \cap \Lambda'_{+} = \emptyset$,
and $E = E_{-} \cup E_+$ with $E_- \cap E_+ = \emptyset$.

Define $\A_{\pm}^{B} = L^{\infty}(G^{\abs{E_{\pm}}})$, and consider
$\A_{\pm}^{B} \subseteq \A^{B}$ as the subalgebra of 
functions 
$F \colon G^{\abs{E}} \rightarrow \C$
that are measurable with respect to $G^{\abs{E_+}}$, that is, functions 
 that depend only on 
the variables $h_{\lra{\lambda\lambda'}}$
with $\lambda$ and $\lambda'$ both in  $\Lambda_{\pm}$.
In this setting, $\A^B$ is the bosonic $q$-double of $\A_{+}^B$.

As in the previous sections, we construct a basis of $\A^{B}$
that is adapted to the reflection.
To find a basis for $L^{\infty}(\Omega_{\{\lambda \lambda'\}}) \simeq L^{\infty}(G)$
with respect to the topology of convergence in measure,
fix a basis $(e_{a})_{a \in S_{\rho}}$ for every irreducible unitary representation 
$(\rho, \cH_{\rho})$
of $G$, and consider the matrix coefficients 
	\be\label{Defn Gauge Field}
	U_{\lambda\lambda'}^{ab;\rho}(h) 
	= \lra{e_{a}, \rho(h_{\lra{\lambda\lambda'}})e_{b}}\;.
	\ee
By the Peter-Weyl Theorem, they constitute 
an orthonormal basis  of 
$L^2(\Omega_{\{\lambda \lambda'\}}, \mu_{H})$,
labelled by $(\rho, a, b) \in \widehat{G}\times S_{\rho} \times S_{\rho}$.
(Since $L^2$-convergence implies convergence in measure, this is sufficient.)
Note that by unitarity of $\rho$, we have 
\be
U_{\lambda\lambda'}^{ab;\rho} = 
\overline{U}_{\lambda'\lambda}^{ba;\rho}\,.
\ee

If we choose a preferred orientation $\lra{\lambda\lambda'}$ of each unoriented bond 
$\{\lambda\lambda'\}$ in $\abs{E_+}$, we obtain an orthonormal basis 
\be\label{eq:bosonbasis}
	U_{I} = \bigotimes_{\{\lambda,\lambda'\} \in \abs{E_+}} U_{\lambda\lambda'}^{ab;\rho}
\ee
of $\A_{+}^{B}$, labelled 
by $I \in \mathcal{I}_{B} = \abs{E_{+}}^{X}$, 
where $X = \bigsqcup_{\rho \in \widehat{G}}S_{\rho} \times S_{\rho}$.
Note that $U_{I_0} = 1$ if $I_0 \in \cI_{B}$ assigns 
to each bond the matrix element $1$ of the trivial representation.

By \eqref{eq:gaugereflection}, the basis elements 
$U^{ab;\rho}_{\lambda\lambda'}$ reflect as
\be \label{eq:reflectU}
\Theta(U^{ab;\rho}_{\lambda\lambda'}) = 
U^{ba;\rho}_{\vartheta(\lambda')\vartheta(\lambda)}\,.
\ee
Since $B_{IJ} = \Theta(U_I)U_J$ is an orthogonal Schauder basis of 
$\A^B$ for the topology of convergence in measure,
any action $S \in \A^{B}$ can be uniquely expressed as
\be\label{eq:matrixc}
	S = \sum_{I,J \in \mathcal{I}} J_{IJ} \Theta(U_I) U_J\,.
\ee
We denote the matrix of coupling constants by $(J_{IJ})_{\mathcal{I}}$.
The submatrix $(J^{0}_{IJ})_{\mathcal{I}\backslash I_{0}}$ of entries with 
$I, J \neq I_0$ is called the matrix of coupling constants across the reflection plane.

\begin{theorem}\label{thm:RPgauge}
Let $\mu$ be the Haar measure on $G^{\abs{E}}$, let
$S \in \A^{B}$ be a reflection-invariant function, and let 
$\mathbb{E}_{\beta S} \colon \A^{B} \rightarrow \C$ be the expectation 
\[
	\mathbb{E}_{\beta S}(A) = \int_{G^{\abs{E}}} \exp(-\beta S) A(h) \mu(dh)
\]
with respect to the (complex) measure $e^{-\beta S} \mu$. Then $\mathbb{E}_{\beta S}$
is reflection positive on $\A^{B}_{+}$ for every $\beta \geq 0$, if and only if the matrix
$(J^{0}_{IJ})_{\mathcal{I}\backslash \{I_0\}}$ of coupling constants 
across the reflection plane is positive semidefinite.\end{theorem}
\begin{proof}
Since $\mu$ is reflection positive by 
Proposition~\ref{prop:primRBclassicalboson},
the result follows from Theorem~\ref{Thm:classicalwithoutintersection}.
\end{proof}

\subsection{Lattice Yang-Mills Theory}\label{Sect:LatticeYMTheory}
For example, consider the Wilson action 
for Yang-Mills theory
$
	S_{YM} = 
	\sum_{P} S^{P}_{YM}
$,
where $P = \lra{\lambda_0\lambda_1\lambda_2\lambda_3}$ is an oriented elementary square 
or `plaquette' in the `coarse' lattice $\Lambda'$, and 
\be\label{eq:YMplaquette}
	S^{P}_{YM} = 
	\sum_{a_0, a_1, a_2, a_3}
	U_{\lambda_0 \lambda_1}^{a_0a_1}U_{\lambda_1\lambda_2}^{a_1a_2}
	U_{\lambda_2\lambda_3}^{a_2a_3}U_{\lambda_3\lambda_0}^{a_3a_0}
\ee
is the trace of the holonomy around $P$.
Here $U_{\lambda_i\lambda_j}^{ab}$ are the matrix elements of 
$h^{\Lambda'}_{\lambda_i\lambda_j}$ with respect to a fixed
unitary irreducible representation $\rho$ of~$G$,  
defined in \eqref{Defn Gauge Field}. 

Cyclic permutations of the four vertices yield the same plaquette
(and the same contribution), and do not contribute to the sum.
Changing the orientation from 
$\lra{\lambda_0\lambda_1\lambda_2\lambda_3}$
to 
$\lra{\lambda_3\lambda_2\lambda_1\lambda_0}$
changes the oriented plaquette, and yields an extra contribution to the sum.
Since $\overline{U}_{\lambda\lambda'}^{ab} = U_{\lambda'\lambda}^{ba}$,
one checks that
this is the complex conjugate of the original contribution. In particular,
$S_{YM}$ is an hermitian element of $\A^{\Lambda}_{B} \subseteq \A_{B}$. 

We now argue that the Wilson action for Yang-Mills theory defines a reflection-positive 
function in the sense of Theorem \ref{thm:RPgauge}.
The idea of our proof is to use the new vertices $\lambda'' \in \Lambda''$ on the
plaquettes, halfway between every pair $\lambda, \lambda' \in \Lambda'$
of neighboring old vertices.  
These are the black vertices in Figure~\ref{fig:Plaquette}.
(Actually, only the extra degrees of freedom on the reflection plane 
are needed, but the other ones are left in for symmetry reasons.)

In order to prove reflection positivity, we 
express $S^{P}_{YM}$ in terms of the basis 
$B_{IJ} = \Theta(U_I)U_J$, cf \eqref{eq:matrixc}.
If $\lambda_{ij}\in \Lambda''$ is the midpoint between $\lambda_i,\lambda_j \in \Lambda'$,
then 
$U_{\lambda_i\lambda_j}^{ab} = 
\sum_{c}U_{\lambda_i\lambda_{ij}}^{ac}
U_{\lambda_{ij}\lambda_j}^{cb}$.
Expanding \eqref{eq:YMplaquette} for a plaquette 
$P = \lra{\lambda_0\lambda_1\lambda_2\lambda_3}$ 
that intersects the reflection plane in 
$\lambda_{01}$ and $\lambda_{23}$
and using \eqref{eq:reflectU},
 we find
\beqs
S^{P}_{YM} &=& 
\sum_{a_{01}a_{23}}\Theta\Big(
\sum_{a_1,a_{12},a_2}
U^{a_{01}a_1}_{\lambda_{01}\lambda_1}
U^{a_1a_{12}}_{\lambda_{1}\lambda_{12}}
U^{a_{12}a_2}_{\lambda_{12}\lambda_2}
U^{a_{2}a_{23}}_{\lambda_{2}\lambda_{23}}
\Big)
\times \\
& &
\sum_{a_1,a_{12},a_2}
U^{a_{01}a_1}_{\lambda_{01}\lambda_1}
U^{a_1a_{12}}_{\lambda_{1}\lambda_{12}}
U^{a_{12}a_2}_{\lambda_{12}\lambda_2}
U^{a_{2}a_{23}}_{\lambda_{2}\lambda_{23}}\,.
\eeqs
From this, we see that the matrix of coupling constants across 
the reflection plane for $S_{YM}$ is positive semidefinite.

\begin{figure}[h!]
\begin{center}
\resizebox{0.42\textwidth}{!}{%
	\begin{picture}(160,140)(70,13)
		\put(0,0){\includegraphics[scale=2]{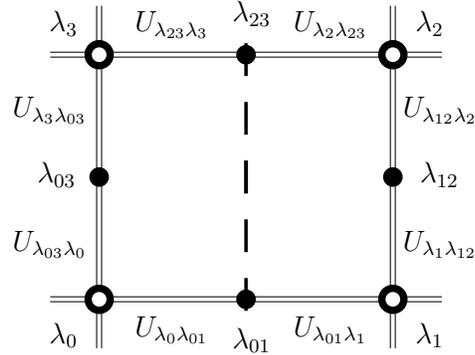}}
		\put(80,12){$\lambda_0$}
		\put(155,10){$\lambda_{01}$}
		\put(230,12){$\lambda_1$}
		\put(232,78){$\lambda_{12}$}
		\put(230,140){$\lambda_2$}
		\put(155,144){$\lambda_{23}$}
		\put(80,140){$\lambda_3$}
		\put(75,78){$\lambda_{03}$}
		\put(180,15){$U_{\lambda_{01}\lambda_{1}}$}
		\put(225,50){$U_{\lambda_{1}\lambda_{12}}$}
		\put(225,105){$U_{\lambda_{12}\lambda_{2}}$}
		\put(180,140){$U_{\lambda_{2}\lambda_{23}}$}
		\put(115,140){$U_{\lambda_{23}\lambda_{3}}$}
		\put(65,50){$U_{\lambda_{03}\lambda_{0}}$}
		\put(65,105){$U_{\lambda_{3}\lambda_{03}}$}
		\put(115,15){$U_{\lambda_{0}\lambda_{01}}$}
	\end{picture}
}
\end{center}
\caption{%
\small Illustration of $S^{P}_{YM}$ for a single plaquette. Note that
$\Theta(U_{\lambda_{01}\lambda_1}) = U^{T}_{\lambda_0\lambda_{01}}$,
\ldots, $\Theta(U_{\lambda_{2}\lambda_{23}}) = U^{T}_{\lambda_{23}\lambda_{3}}$. 
}\label{fig:Plaquette}
\end{figure}

\begin{corollary}\label{cer:rpwilson}
If $\mu$ is the Haar measure on $G^{\abs{E}}$, then the expectation
$\mathbb{E}_{YM} \colon \A_{B} \rightarrow \C$ defined by
\[
 \mathbb{E}_{YM}(A) = \int_{G^{\abs{E}}}
\exp(-{\textstyle \frac{1}{2g_0^2}}S_{YM})A(h)  \;\mu(dh)
\]
is reflection positive on $\A^+_{B}$ for all values of $g_0$.
\end{corollary}
\begin{proof}
This follows from 
Theorem~\ref{thm:RPgauge}
by the previous discussion.
\end{proof}
In particular, the expectation $\mathbb{E}_{YM}$ is reflection positive 
on the subalgebra $\A^{B}_{\Lambda'}$ of functions that depend only 
on the bond variables $h^{\Lambda}_{\lambda'\kappa'}$ between 
points $\lambda', \kappa'$ in the `coarse' lattice $\Lambda'$.

Note that since our derivation does not use gauge invariance, we get reflection 
positivity of the full algebra, not just the gauge invariant part.

\subsection{Fermions in Lattice Gauge Theory}

The fermionic degrees of freedom live only on the `coarse' sublattice $\Lambda'$, 
(the white dots in Figure~\ref{fig:QCDlattice}).

To a single site $\lambda\in \Lambda'$, we associate the Grassmann 
algebra $\A^{F}_{\lambda} = \bigwedge V$.
Here $V = W \oplus W^*$, where  
$W = W_{s}\otimes W_{\rho}$ is the tensor product 
of a $\mathrm{Cl}(\R^4)$-representation $W_{s}$ and a  
unitary $G$-representation $W_{\rho}$.
Both $G$ and $\mathrm{Cl}(\R^4)$ act from the left on $W$, and from the right 
on $W^*$.

The algebra of observables for the fermionic part of the theory is 
\be\label{eq:grassmanQCD}
	\A^{F} = \bigwedge \bigoplus_{\lambda\in \Lambda}(W\oplus W^*)\,,
\ee
and the full algebra of observables is $\A = \A^{B}\otimes \A^{F}$.

Choose a basis $\psi_{\alpha a}$ of $W$, and denote the dual basis 
of $W^*$ by $\overline{\psi}_{\alpha a}$.
The map $\psi_{\alpha a} \mapsto \overline{\psi}_{\alpha a}$
extends to an antilinear isomorphism $\psi \mapsto \overline{\psi}$
from $W$ to $W^*$.
From the basis of $W \oplus W^*$,
we obtain 
anticommuting generators
$\psi_{\alpha a}$ and $\overline{\psi}_{\alpha a}$ of~$\A^{F}_{\lambda}$.
Using these, we find
anticommuting generators $\psi_{\alpha a}(\lambda)$ and 
$\overline{\psi}_{\alpha a}(\lambda')$ of $\A^{F}$,
\[
\{\psi_{\alpha a}(\lambda), \psi_{\alpha' a'}(\lambda')\} = 
\{\psi_{\alpha a}(\lambda), \overline{\psi}_{\alpha' a'}(\lambda')\} = 
\{\overline{\psi}_{\alpha a}(\lambda), \overline{\psi}_{\alpha' a'}(\lambda')\} = 0\,.
\]

\subsubsection{The Reflection}
Assume that $\vartheta \colon \Lambda \rightarrow \Lambda$ 
flips a single coordinate $x^{\sigma}$, and that the fixed point set $P$
is as in \S\ref{sec:RPgluons}. 
Then the corresponding reflection 
$\Theta \colon \A \rightarrow \A$ is the unique antilinear homomorphism 
satisfying 
\beq
	\Theta \overline{\psi}_{\lambda} &=& -i \gamma_{\sigma} \psi_{\vartheta(\lambda)}\\
	\Theta \psi_{\lambda} &=& -i \overline{\psi}_{\vartheta(\lambda)}\gamma_{\sigma}\\
	\Theta(F)(h_{\lra{\lambda,\lambda'}}) &=& \overline{
		     F(h_{\lra{\vartheta(\lambda),\vartheta(\lambda')}})
		}
\eeq
for all $F \in \A^{B}$ and $\psi \in W$, $\overline{\psi} \in \overline{W}$.
Here, the $\gamma_{\mu}$ are \emph{euclidean} Dirac matrices
satisfying $\{\gamma_{\mu}, \gamma_{\nu}\} = 2\delta_{\mu\nu}$
and $\gamma_{\mu}^{\dagger} = \gamma_{\mu}$. 

\begin{remark}
Note that we require $\Theta$ to be an antilinear \emph{homomorphism},
satisfying
$\Theta(AB) = \Theta(A)\Theta(B)$.
This deviates slightly from e.g.\ \cite{OS78, S82}, where 
an antilinear \emph{anti-homomorphism} $\Theta_a$ is used, satisfying
$\Theta_{a}(AB) = \Theta_{a}(B)\Theta_{a}(A)$.
For super-commutative algebras such as $\A$, one checks that 
homomorphisms are related to anti-homomorphisms by
$\Theta_{a}(A) = i^{\abs{A}^2}\Theta(A)$ for homogeneous $A \in \A$,
cf. Remark~\ref{rk:switch}.
\end{remark}

The algebra
$\A_+$ is defined as $\A_+ = \A^{B}_{+} \otimes \A^{F}_+$.  Here
$\A_{+}^{B} = L^{\infty}(G^{\abs{E_{+}}})$ as before, and 
 $\A^{F}_{+}$ is the Grassmann algebra
\[
	\A_+^F = \bigwedge 
	\bigoplus_{\lambda \in \Lambda'_+ } (W \oplus W^*)\,.
\]

As usual, we use a basis of $\A_+$ to construct a basis of the even 
subalgebra $\A^0 \subseteq \A$ 
that is well adapted to the reflection.
Recall that $\A^{B}_{+}$ has the basis $U_{I_{B}}$ described in 
equation~\eqref{eq:bosonbasis},
labelled by the set $\mathcal{I}_{B} = \abs{E_{+}}^{X}$.
A basis 
\[\Psi_{I_{F}} = \psi_{\alpha_1 a_1}(\lambda_1) \wedge \cdots \wedge 
\psi_{\alpha_{k}a_{k}}(\lambda_k)\] 
of $\A_{+}^{F}$ can be constructed
as in Section~\ref{sec:fermpsi}. 
Since the fermions only live on the `coarse' lattice $\Lambda'$
which does not intersect the fixed point set,
this basis is labelled by  
$I_{F}$  in 
$\mathcal{I}_{F} = \mathcal{P}(T\sqcup \overline{T} \times \Lambda'_{+})$,
where $T$ is the set of labels $(\alpha,a)$ of basis vectors of $W$. 
%
%

Finally, we obtain a basis $C_{I} = U_{I_B} \otimes \Psi_{I_F}$ of $\A_{+}$, 
labelled by $\mathcal{I} = \mathcal{I}_{B} \times \mathcal{I}_{F}$.
The identity is labelled by $I_0 = (I^B_{0}, I^{F}_{0})$, where $I^{B}_{0}$
labels the identity as before,
and $I^{F}_{0} = \emptyset$.

If we set $B_{IJ} = i^{\abs{C_{I}}^2}\Theta(C_I)C_J$ for basis elements 
$C_I$ and $C_J$ of the same $\Z_2$-degree, then 
any action $S \in \A^{0}$ of degree zero can be uniquely expressed as
\be\label{eq:matrixc2}
	S = \sum_{I,J \in \mathcal{I}} J_{IJ} B_{IJ}\,.
\ee
We denote the matrix of coupling constants by $(J_{IJ})_{\mathcal{I}}$.
The submatrix $(J^{0}_{IJ})_{\mathcal{I}\backslash I_{0}}$ of entries with 
$I, J \neq I_0$ is called the matrix of coupling constants across the reflection plane.

Let $\tau \colon \A \rightarrow \C$ be the tensor product of the Berezin integral
$\tau_{F} \colon \A_{F} \rightarrow \C$ of Section~\ref{sec:fermpsi}
and the expectation $\mathbb{E} \colon \A_{B} \rightarrow \C$
with respect to the Haar measure on $G^{\abs{E}}$.

\begin{theorem}\label{thm:RPgauge2}
Let $S \in \A$ be a reflection-invariant action of degree zero. 
Then the functional $\tau_{S}(A) = \tau(e^{- \beta S} A)$
is reflection positive for every $\beta \geq 0$, if and only if the matrix
$(J^{0}_{IJ})_{\mathcal{I}\backslash \{I_0\}}$ of coupling constants 
across the reflection plane is positive semidefinite.\end{theorem}
\begin{proof}
By Proposition~\ref{prop:primfunc} with $\overline{\rho} = i \gamma_{\tau}$, the 
continuous functional
$\tau_{F}$ is factorizing and reflection positive.
By Proposition~\ref{prop:primRBclassicalboson},
the same is true for $\mathbb{E}$, hence also for the functional
$\tau \colon \A \rightarrow \C$.
The result then follows from Theorem~\ref{Thm:neccandsuff}.
\end{proof}

\subsection{Lattice QCD}

We apply this theorem to the lattice QCD-action
$S = S_{YM} + S_F$. Here, the fermion action $S_F = S_{FM} + S_{FK}$ is the sum of
a mass term and a kinetic term,
\beq
		S_{FM} &=& \frac{1}{2}\sum_{\lambda \in \Lambda'} \overline{\psi}_{a\alpha}(\lambda)
	\Gamma^{\alpha\beta} \psi_{a\beta}(\lambda)\,,\label{eq:massterm}\\
	S_{FK} &=& 
	\frac{\kappa}{2} \sum_{\lra{\lambda\lambda'} \in E_{\Lambda'}}
	\overline{\psi}_{\alpha a}(\lambda)
	\Gamma^{\alpha\beta}_{\lambda' - \lambda}
	U_{\lambda\lambda'}^{ab}\psi_{\beta b}(\lambda')\,. \label{eq:kineticterm}
\eeq
The first sum is over sites $\lambda$ in the `coarse' lattice $\Lambda'$, 
and the second sum is over all \emph{oriented} nearest neighbor bonds in $\Lambda'$.
(So every pair gives two contributions.)
Recall that  
$U_{\lambda\lambda'}^{ab} = \sum_{c}U_{\lambda\lambda''}^{ac}U_{\lambda''\lambda'}^{cb}$ for the site
$\lambda''\in \Lambda''$ halfway in between 
$\lambda,\lambda' \in \Lambda'$.

We prove reflection positivity for couplings
\be\label{eq:parameters}
	\Gamma = (M-4s)\one,\quad 
	\Gamma_{\lambda' - \lambda} = \pm\gamma_{\mu} + s\one
	\quad \text{if}\quad \lambda' -\lambda = \pm r \vec{e}_{\mu}\,,
\ee
where $s=0$ or $s=1$. The choice $s=0$ corresponds to the `naive' action
(which leads to fermion doubling in the continuum limit), and the choice 
$s=1$ corresponds to Wilson's action.



Note that $S_{F}$ is reflection symmetric, $\Theta(S_{F}) = S_{F}$.
Indeed, a straightforward calculation shows that this follows from 
$(\gamma_{\sigma} \overline{\Gamma} \gamma_{\sigma})^T = \Gamma$
for the mass terms, and from
$\Gamma_{-\vartheta(\lambda' - \lambda)} = 
(\gamma_{\sigma}\overline{\Gamma}_{\lambda' - \lambda}\gamma_{\sigma})^{T}$
for the kinetic terms. 

By Theorem~\ref{thm:RPgauge2}, the mass terms in \eqref{eq:massterm} 
are irrelevant, as they only contain terms in either $\A_+$ or $\A_{-}$.
The same holds for the kinetic 
terms with both $\lambda$ and $\lambda'$ in either $\Lambda'_{+}$
or $\Lambda'_{-}$.
Therefore, it suffices to consider terms of the form
\be\label{eq:term}
\sum_{\alpha,\beta,a,b,c}
\overline{\psi}_{\alpha a}(\lambda_-)
	\Gamma^{\alpha\beta}_{\lambda_+ - \lambda_-}
	U_{\lambda_-\lambda_0}^{ac} U_{\lambda_0\lambda_+}^{cb}\psi_{\beta b}(\lambda_+)\,,
\ee
with $\lambda_{0} \in \Lambda_0$, and either 
$\lambda_{\pm} = \lambda_0 \pm \frac{1}{2}r \vec{e}_{\sigma}$
or $\lambda_{\pm} = \lambda_0 \mp \frac{1}{2}r \vec{e}_{\sigma}$.
Note that
\[
	\Theta\Big(\sum_{b}U_{\lambda_0\lambda_+}^{cb}\psi_{\beta b}(\lambda_+)\Big)
	= -i \sum_{a} \overline{\psi}_{a\alpha}(\lambda_-)  \gamma_{\sigma}^{\alpha\beta}U_{\lambda_-\lambda_0}^{ac}\,.
\]
If $\lambda_{\pm} \in \Lambda_{\pm}$ and 
$\Gamma_{\lambda_{+} - \lambda_-} = s \one + \gamma_{\sigma}$, then
the expression \eqref{eq:term} can be written as $i \Theta(X_{\alpha}) X_{\alpha}$, where
$X_{\alpha}\in \A^+$ is given by
\[
	X_{\alpha} = \kappa_{\beta\alpha}U_{\lambda_0\lambda_+}^{cb}\psi_{\beta b}(\lambda_+)\,,
\quad\text{with}\quad 
\kappa = {\textstyle\frac{1}{\sqrt{1+s}}}(\one + s\gamma_{\sigma})\,.
\]
If $\lambda_{\pm} \in \Lambda_{\mp}$ and 
$\Gamma_{\lambda_{+} - \lambda_-} = s \one - \gamma_{\sigma}$, then
\eqref{eq:term} can be written as $i X_{\alpha}\Theta(X_{\alpha})$, where
$X_{\alpha}\in \A^-$ is given by
\[
	X_{\alpha} = \kappa_{\beta\alpha}U_{\lambda_0\lambda_+}^{cb}\psi_{\beta b}(\lambda_+)\,,
\quad\text{with}\quad 
\kappa = {\textstyle\frac{1}{\sqrt{1+s}}}(\one - s\gamma_{\sigma})\,.
\]
From this, one concludes that the matrix of coupling constants across the reflection plane 
is positive semidefinite.

\begin{theorem}
For the lattice QCD Lagrangian 
$S = S_{YM} + S_{F}$, the linear functional
$\tau_{S} \colon \A \rightarrow \C$ defined by
$A \mapsto \tau(\exp(-S)A)$ is reflection positive with respect 
to $\A_{+}$
for $s\in \{0,1\}$, for all $M, g_{0}\in \R$, and for all $\kappa \geq 0$.
\end{theorem}

Although this theorem holds for  reflections 
in each of the four coordinate directions, the 
physical Hilbert space $\cH_{\Theta}$ is derived from reflections 
in the time direction $x^0$. It is
the completion of $\A_{+}$
with respect to the positive semidefinite inner product
\[
	\lra{A_+,B_+}_{\Theta} = \tau(e^{-S} \Theta(A_+)B_+)\,.
\]

\subsection{Gauge Transformations}

Denote the 4d-gauge group by $G^{\Lambda}$, and 
the 3d-gauge group by $G^{\Lambda_{0}}$.
We identify $G^{\Lambda_0}$ with the quotient of 
$G^{\Lambda}$ by the normal subgroup 
\[N = \{g\in G^{\Lambda}\,;\, g|_{\Lambda_0} = \one|_{\Lambda_0}\}\,.\]

Every $g\in G^{\Lambda}$ induces an automorphism of $\A$, namely the unique one 
satisfying $h_{\lambda\lambda'} \mapsto g_{\lambda} h_{\lambda\lambda'}g^{-1}_{\lambda'}$, 
$\psi_{\lambda} \mapsto \rho(g_{\lambda})\psi_{\lambda}$, and
$\overline{\psi}_{\lambda} \mapsto \overline{\psi}_{\lambda}\rho(g^{-1}_{\lambda})$.
(The fermions transform under a unitary representation $\rho$ of $G$.)
Note that $\Theta \alpha_{g}\Theta = \alpha_{\theta(g)}$, with
$\theta(g)_{\lambda} = g_{\vartheta(\lambda)}$.

\begin{proposition}
This yields a unitary representation of $G^{\Lambda}$ on $\cH_{\Theta}$,
which factors through the quotient $G^{\Lambda_0} \simeq G^{\Lambda}/ N$.
\end{proposition}
\begin{proof}
Since $\alpha_{g}$ maps $\A_+$ to $\A_+$, we have
\beqs
\lra{\alpha_g (A_+), \alpha_{g}(B_+)}_{\Theta} &=& 
i^{\abs{A_+}^2}\tau(e^{-S} \Theta(\alpha_g (A_+))\alpha_g(B_+) )\\
&=& i^{\abs{A_+}^2}\tau(e^{-S} \alpha_{\theta(g)}(\Theta(A_+))\alpha_g(B_+) )\\
&=& i^{\abs{A_+}^2}\tau( e^{-S} \alpha_{\overline{g}}(\Theta(A_+)(B_+) )\,.
\eeqs
Here $\overline{g}$ is the gauge transformation with $\overline{g}|_{\Lambda_+} = g|_{\Lambda_+}$
and $\overline{g}|_{\Lambda_-} = \theta(g)|_{\Lambda_-}$. Since both $S$ and $\tau$
are gauge invariant, this equals $\lra{A_+,B_+}_{\Theta}$.
It follows that the null space of the positive semidefinite form is gauge invariant,
and that $G^{\Lambda}$ acts unitarily on $\cH_{\Theta}$. 

We show that $g$ acts trivially if $g|_{\Lambda_0} = \one|_{\Lambda_0}$.
For this, note that 
\[
\|\alpha_{g}(A_+) - A_+\|^2_{\Theta} = 2\lra{A_+,A_+}_{\Theta} - 2\mathrm{Re} \lra{A_+,\alpha_g(A_+)}_{\Theta}\,.
\]
If $g|_{\Lambda_0}$ is trivial, then
$\lra{A_+,\alpha_g(A_+)}_{\Theta} = \tau(e^{-S}\Theta(A_+)\alpha_{g}(A_+))$ is equal to 
$\tau(e^{-S}\alpha_{g_+}(\Theta(A_+)A_+)) = \lra{A_+,A_+}_{\Theta}$, 
with $g_+|_{\Lambda_+} = g|_{\Lambda_+}$
and $g|_{\Lambda_-} = \one|_{\Lambda_-}$. 
It follows that $\alpha_{g}$ acts trivally on $\cH_{\Theta}$ for $g\in N$, 
so the representation factors through the quotient $G^{\Lambda_0} \simeq G^{\Lambda} / N$.
\end{proof}

\begin{remark}
In particular, we retain an action of the global gauge group $G$, which sits inside 
$G^{\Lambda_0}$ as the group of constant $G$-valued functions.
This allows one to define charge operators on $\cH_{\Theta}$.  
\end{remark}


\setcounter{equation}{0}
\section{Parafermions}\label{sec:parafermionsNo2}

We characterize reflection positivity for parafermions.
Here, we need our lattice $\Lambda$
to be \emph{ordered}, and the reflection $\vartheta \colon \Lambda \rightarrow \Lambda$
to be \emph{order reversing} and \emph{fixed point free}.
We allow $\Lambda$ to be either finite or countably infinite, and we define 
$\Lambda_+\subseteq \Lambda$ as the maximal subset 
with $\vartheta(\Lambda_+) < \Lambda_+$.
The CPR algebra $\A(q,\Lambda)$, considered in 
\S\ref{sec:ParaCPR}, is then
 the $q$-double of the CPR algebra $\A(q,\Lambda_{+})$.
The `background functional' is the tracial state 
$\tau \colon \A(q,\Lambda) \rightarrow \C$ of 
Proposition~\ref{prop:paratraceisnice}. 

The operators $C_{I}$ of equation \eqref{BasisPF}, labelled by
$I \in \Z^{\Lambda_+}_{p}$, constitute a homogeneous Schauder 
basis of $\A(q,\Lambda_+)$ with respect to the norm topology,
satisfying B1--3.
We can thus form a basis $B_{IJ}$ of $\A^{0}(q,\Lambda)$ by
\[
	B_{IJ} = \Theta(C_{I})\circ C_{J}= \zeta^{\abs{I}^2}\Theta(C_I)C_J\,,
\]
labelled by $I, J \in \Z^{\Lambda_{+}}_{p}$ with $\abs{I} = \abs{J} \in \Z_p$.
Here $\abs{I} =\sum_{\Lambda_+} I_{\lambda}$ denotes
the degree of $C_I$ in $\Z_p$.  
Any element $H \in \A$ of degree zero therefore has a unique norm convergent expansion
\be\label{eq:parahamsum}
	-H = \sum_{I,J \in \Z^{\Lambda_{+}}_{p}} J_{IJ} B_{IJ}\,,
\ee
with \emph{coupling matrix} $(J_{IJ})_{\Z_{p}^{\Lambda_+}}$. 
Denote by $(J^0_{IJ})_{\Z_{p}^{\Lambda_+} \backslash \{0\}}$ the matrix of
\emph{couplings across the reflection plane}, namely
the submatrix of entries with with ${I,J \neq 0}$.

\begin{theorem}
Let $H \in \A(q,\Lambda)$ be a reflection invariant operator of degree zero.
Then the functional 
$\tau_{\beta H}(A) = \tau(A\,e^{-\beta H} )$ is reflection positive on 
$\A(q,\Lambda_+)$
for all $\beta \geq 0$ if and only if the matrix 
$(J^0_{IJ})_{\Z_{p}^{\Lambda_+} \backslash \{0\}}$
of coupling constants across the reflection plane is positive semidefinite.
\end{theorem}
\begin{proof}
This follows 
from Proposition~\ref{prop:paratraceisnice} and Theorem~\ref{Thm:neccandsuff}.
\end{proof}

\begin{remark}
If the lattice is infinite, then the expression \eqref{eq:parahamsum}
for the Hamiltonian $H$ is usually 
not a norm convergent sum.
One then approximates these expressions by a sequence $H_{N}$
of convergent Hamiltonians in $\A(q,\Lambda)$.
\end{remark}

\goodbreak

\section*{Acknowledgments}
\noindent 
A.J. was supported in part by a grant ``On the Mathematical Nature of the Universe'' from the Templeton Religion Trust.  B.J.\ was supported by the NWO grant 613.001.214 ``Generalised Lie algebra sheaves." He  thanks A.J.\ for hospitality at Harvard University.  Both authors thank the Hausdorff Institute for Mathematics and the Max Planck Institute for Mathematics in Bonn for hospitality during part of this work.  

\bibliographystyle{alpha}

\end{document}